\documentclass[journal]{IEEEtran}
\usepackage{booktabs} 
\usepackage{paralist}
\usepackage[formats]{listings}
\usepackage{algorithm}
\usepackage[noend]{algpseudocode}
\usepackage{amssymb,amsmath,amsthm,amsfonts}
\usepackage{mathtools}
\usepackage{framed}
\newtheorem*{remark}{Remark}
\usepackage{multirow}
\usepackage{caption}
\usepackage{subcaption}
\usepackage{color}
\newtheorem{definition}{Definition}
\newtheorem{theorem}{Theorem}
\ifCLASSINFOpdf
\else
\fi

\begin{document}
	%
	\title{CHOP: Bypassing Runtime Bounds Checking Through Convex Hull OPtimization}

	
	\author{\IEEEauthorblockN{Yurong Chen,
			Hongfa Xue,
			Tian Lan,
			Guru Venkataramani}
}

	
	%



	\IEEEtitleabstractindextext{%
		
		\begin{abstract}
			Unsafe memory accesses in programs written using popular programming languages like C/C++ have been among the leading causes for software vulnerability. Prior memory safety checkers such as SoftBound enforce memory spatial safety by checking if every access to array elements are within the corresponding array bounds. However, it often results in high execution time overhead due to the cost of executing the instructions associated with bounds checking. To mitigate this problem, redundant bounds check elimination techniques are needed. In this paper, we propose CHOP, a Convex Hull OPtimization based framework, for bypassing redundant memory bounds checking via profile-guided inferences. In contrast to existing check elimination techniques that are limited by static code analysis, our solution leverages a model-based inference to identify redundant bounds checking based on runtime data from past program executions. For a given function, it rapidly derives and updates a knowledge base containing sufficient conditions for identifying redundant array bounds checking. We evaluate CHOP on real-world applications and benchmark (such as SPEC) and the experimental results show that on average 80.12\% of dynamic bounds check instructions can be avoided, resulting in improved performance up to 95.80\% over SoftBound.
		\end{abstract}
		
		}

	\maketitle

	\IEEEdisplaynontitleabstractindextext

	%
	\IEEEpeerreviewmaketitle

\section{Introduction}
\label{sec:intro}

Many software bugs and vulnerabilities in C/C++ applications occur due to the unsafe pointer usage and out-of-bound array accesses. 
This also gives rise to security exploits taking advantage of buffer overflows or illegal memory reads and writes. Below are some of the recent examples.
\begin{inparaenum}[i)]
	\item A stack overflow bug inside function $getaddrinfo()$ from glibc was discovered by a Google engineer in February 2016. Software using this function could be exploited with attacker-controller domain names, attacker-controlled DNS servers or through man-in-the-middle attacks~\cite{glibc}.
	\item Cisco released severe security patches in 2016 to fix a buffer overflow vulnerability in the Internet Key Exchange (IKE) from Cisco ASA Software. This vulnerability could allow an attacker to cause a reload of the affected system or to remotely execute code~\cite{Cisco}. 
\end{inparaenum}

In order to protect software from spatial memory/array bounds violations, tools such as SoftBound~\cite{SoftBound_PLDI_2009} have been developed that maintains metadata (such as array boundaries) along with rules for metadata propagation when loading/storing pointer values. By doing so, SoftBound ensures that pointer accesses do not violate boundaries by  performing runtime checks. While such a tool offers protection from spatial safety violations in programs, we should also note that they often incur high performance overheads due to a number of reasons.
\begin{inparaenum}[a)]
	\item Array bounds checking add extra instructions in the form of memory loads/stores for pointer metadata, which also needs to be duplicated and passed between pointers during assignments.
	\item In pointer-intensive programs, such additional memory accesses can introduce memory bandwidth bottleneck and further degrade system performance.
\end{inparaenum} 

To mitigate runtime overheads, static techniques to remove redundant checks have been proposed. ABCD ~\cite{ABCD_SIGPLAN_2000} builds and solves systems of linear inequalities involving array bounds and index variables, while WPBound~\cite{wpcheck} statically computes the potential ranges of target pointer values inside loops, then compares them with the array bounds obtained from SoftBound to avoid SoftBound-related checks.

\begin{figure}[t]
	\begin{center}
		\includegraphics[scale=0.5]{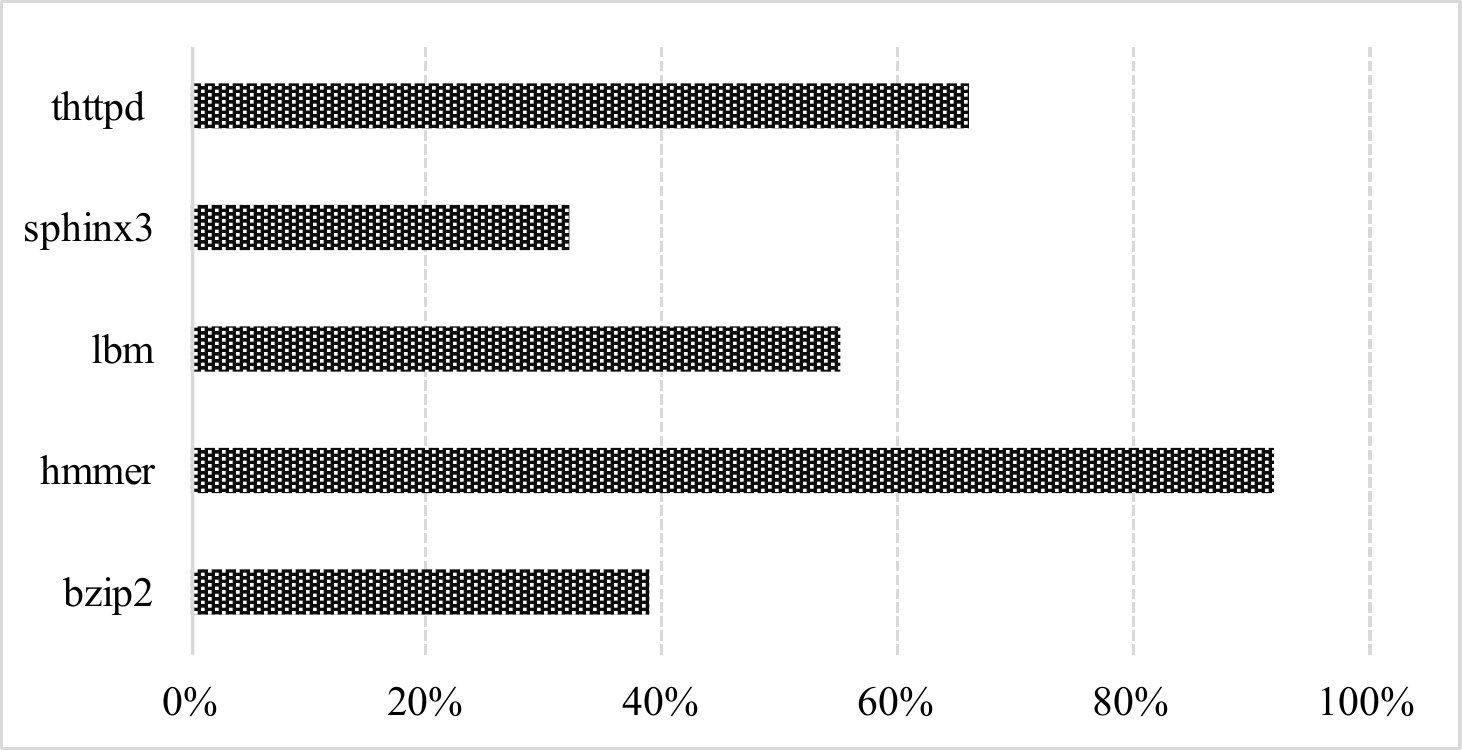}
		\caption{Runtime overhead for SoftBound compared to original application}
		\label{fig:RuntimeOverhead}
	\end{center}
	\vspace{-0.3in}
\end{figure}

However, such static approaches are limited by a tradeoff between the tractability of static analysis and the effectiveness of redundant checks identification, because optimally removing redundant checks may require building and solving constraint systems that become prohibitive. For programs at-scale, static analysis is often restricted to considering simplified constraint systems (e.g., only difference constraints in~\cite{ABCD_SIGPLAN_2000}) and thus falls short on achieving high redundant-check removal/bypassing rate.

\begin{figure*}[t]
	\centering
	\begin{minipage}{.31\textwidth}
		\captionsetup[lstlisting]{font={small,tt}}
		\caption{Non-instrumented Code}
		\label{fig:originalCode}
		\begin{lstlisting}[xleftmargin=0em,numbers=left,frame = single ,basicstyle=\scriptsize\ttfamily]
		static void
		foo(char* src, char* dst, 
		int ssize, int dsize, int snum)
		{
		char* cp1;
		char* cp2;
		if(ssize+3*snum+1>dsize){
		dsize = ssize+3*snum;
		dst = (char*) realloc(dst,dsize);
		}
		for ( cp1 = src, cp2 = dst;
		*cp1 != '\0' && 
		cp2 - dst < dsize - 1;
		++cp1, ++cp2 )
		{
		
		switch ( *cp1 )
		{
		case '<':
		*cp2++ = '&';
		*cp2++ = 'l';
		*cp2++ = 't';
		*cp2 = ';';
		break;
		case '>':
		*cp2++ = '&';
		*cp2++ = 'g';
		*cp2++ = 't';
		*cp2 = ';';
		break;
		default:
		*cp2 = *cp1;
		break;
		}
		}
		
		*cp2 = '\0';
		}
		\end{lstlisting}
	\end{minipage}
	\begin{minipage}{.33\textwidth}
		\captionsetup[lstlisting]{font={small,tt}}
		\caption{SoftBound Instrumented Code}
		\label{fig:SBCode}
		\begin{lstlisting}[xleftmargin=0em,numbers=none,frame = single ,basicstyle=\scriptsize\ttfamily]
		static void
		foo_SB(char* src, char* dst,
		int ssize, int dsize, int snum)
		{
		char* cp1; char* cp2;
		if(ssize+3*snum+1>dsize){
		dsize = ssize+3*snum;
		dst = (char*) realloc(dst,dsize);
		}
		for ( cp1 = src, cp2 = dst;
		*cp1 != '\0' && cp2 - dst < dsize - 1;
		++cp1, ++cp2 )
		{
		switch ( *cp1 )
		{
		case '<':
		//CHOP: trip count tc1 here
		CHECK_SB(cp2);*cp2++ = '&';
		CHECK_SB(cp2);*cp2++ = 'l';
		CHECK_SB(cp2);*cp2++ = 't';
		CHECK_SB(cp2);*cp2 = ';';
		break;
		case '>':
		//CHOP: trip count tc2 here
		CHECK_SB(cp2);*cp2++ = '&';
		CHECK_SB(cp2);*cp2++ = 'g';
		CHECK_SB(cp2);*cp2++ = 't';
		CHECK_SB(cp2);*cp2 = ';';
		break;
		default:
		//CHOP: trip count tc3 here
		CHECK_SB(cp1);CHECK_SB(cp2);
		*cp2 = *cp1;
		break;
		}
		}
		CHECK_SB(cp2);*cp2 = '\0';
		}
		\end{lstlisting}
	\end{minipage}
	\begin{minipage}{.34\textwidth}
		\captionsetup[lstlisting]{font={small,tt}}
		\caption{CHOP Optimized Code}
		\label{fig:CHOPCode}
		\begin{lstlisting}[xleftmargin=0em,numbers=none,frame = single ,basicstyle=\scriptsize\ttfamily]
		//original foo() function 
		static void
		foo
		(char* src, char* dst, 
		int ssize, int dsize, int snum)
		{...}
		
		
		//SoftBound instrumented foo()
		static void
		foo_SB
		(char* src, char* dst, 
		int ssize, int dsize, int snum)
		{...}
		
		
		
		int
		main()
		{
		char *src, *dst;
		int ssize, dsize, snum;
		...  
		/*determine whether it's
		inside the safe region*/
		
		if(CHECK_CHOP(src,dst,ssize,dsize,snum)) 
		{
		foo(src,dst,ssize,dsize,snum);
		}
		else
		{
		foo_SB (src,dst,ssize,dsize,snum)
		} 
		
		...
		
		}   
		\end{lstlisting}
	\end{minipage}
	\vspace{-0.3in}
\end{figure*}

In this paper, we propose CHOP, a novel approach that builds and verifies conditions for eliminating bounds checking on the fly by harnessing runtime information instead of having to rely on discovering redundant checks solely during compile-time or using static code analysis. CHOP is effective in bypassing a vast majority of redundant array checks while being simple and elegant. The key idea is to infer the safety of a pointer dereference based on statistics from past program executions. If prior executions show that the access of array {\it A} with length {\it L} at index {\it i} is within bound (which is referred to as a \textbf{data point}), then it is safe to remove the checks on any future access of A with length no smaller than {\it L} and an index no larger than {\it i}. 
As a result, a \textbf{``safe region"} is built by combining the ranges derived from relevant variables and array lengths in past executions. Any future dereference of the target pointer will be regarded as safe if it falls within the safe region. In general, a safe region is the area that is inferred and built from given data points, such that for any input points within the region, the corresponding target pointer is guaranteed to have only safe memory access, e.g., all bounds checking related to the pointer can be removed. We investigated two methods to effectively construct safe regions, i.e., the union and convex hull approaches. The union approach builds a safe region by directly merging the safe regions that are defined by each individual data point. While the union approach is light-weight and sufficient for data points with low dimensions, it does not have the ability to infer a larger safe region from known data points (e.g., through an affine extension), which is crucial for high-dimensional data points. 
In such cases, we can further expand the union of safe regions to include the entire convex hull, which is the smallest convex set containing all known data pointers and their safe regions. Due to such inference, our convex hull approach is able to render a higher ratio of redundant check bypassing. As demonstrated through function $defang()$ from $thttpd$ application, the convex hull approach is shown to achieve 82.12\% redundant check bypassing compared with 59.53\% in union approach. To further improve efficiency, we prioritize CHOP to bounds-check performance hotspots that incur highest overhead with SoftBound.


In this article, we make the following significant contributions compared to our previous work SIMBER~\cite{SIMBER}:
\begin{asparaenum}
	\item We propose CHOP, a tool that let programs bypasses bounds checking by utilizing convex hull optimization and runtime profile-guided inferences. We utilize a convex hull-based approach to build the safe regions for pointer accesses. With convex hull optimization, CHOP can efficiently handle high-dimensional data points and the runtime bounds check bypassing ratio is improved against SIMBER.
	\item We observed no \textbf{false positives} of bounds check bypassing from our experimental results. CHOP identifies a bounds check as redundant only if it is deemed unnecessary using the sufficient conditions derived from past program executions. (A ``false positive'' means a bounds check that should be conducted is wrongly bypassed.)
	\item We evaluate CHOP on expanded set of real-world benchmarks and validate significant overhead reduction of spatial safety checks by 66.31\% compared to SoftBound on average. 
\end{asparaenum}


\section{System Overview}
\label{sec:background}

SoftBound stores the pointer metadata (array base and bound) when pointers are initialized, and performs array bounds checking (or validation) when pointers are dereferenced. For example, for an integer pointer {\it ptr} to an integer array {\it intArray}[100], SoftBound stores {\it ptr\_base} = \&{\it intArray}[0] and {\it ptr\_bound} = {\it ptr\_base} + {\it size(intArray)}. When dereferencing {\it ptr+offset}, SoftBound obtains the base and bound information associated with pointer {\it ptr}, and performs the following check: if the value of {\it ptr} is less than {\it ptr\_base}, or, if {\it ptr+offset} is larger than {\it ptr\_bound}, the program terminates.

A disadvantage for such an approach is that, it can add performance overheads to application runtime especially due to unnecessary metadata tracking and pointer checking for benign pointers. Fig.~\ref{fig:RuntimeOverhead} shows the runtime overhead of SoftBound instrumented applications over original applications, taking thttpd and SPEC2006~\cite{spec} as benchmarks. Existing works~\cite{ABCD_SIGPLAN_2000,wpcheck} mainly analyze relationship between variables in source code, build constraint systems based on static analysis and solve the constraints to determine redundant checks.

In CHOP, we propose a novel framework where the bounds check decisions are made using runtime data and inferences. Our results show that even limited runtime data can be quite powerful in inferring the safety of pointer dereferences. Consider the example shown in Fig.~\ref{fig:originalCode}, where $foo(src, dst, ssize, dsize, snum)$ converts the special characters `$<$' and `$>$' in string $src$ of length $ssize$ into an HTML expression while keeping other characters unchanged. The result is stored in $dst$ of length $dsize$. The total number of special characters is $snum$. Pointer $cp2$ is dereferenced repeatedly inside the $for$ loop, e.g., in lines 20-23 and 26-29. If SoftBound is employed to ensure memory safety, bounds checking (denoted by $CHECK\_SB$ in Fig.~\ref{fig:SBCode}) will be added before each pointer dereference. For every iteration of the $for$ loop, the {\it CHECK\_SB} will be executed, thus leading to intensive checks and overhead. We note that a buffer overflow will occur only if $cp2$ is smaller than $dst+dsize-1$ at the end of the second last iteration of the $for$ loop, but exceeds $dst+dsize-1$ during the last iteration. Later, when $cp2$ is dereferenced, the access of string $dst$ is past by the bound given by $dsize$. It is easy to see the number of iterations visiting line 19, 25 and 31 determines exactly the length of string $dst$. Therefore, $dst$ will have final length $4*(tc1+tc2)+tc3+1$, where $tc1$, $tc2$ and $tc3$ are three auxiliary branch-count variables instrumented by CHOP. Any bounds check can be safely removed as long as $4*(tc1+tc2)+tc3+1\le dsize$.

Existing static approaches such as ABCD~\cite{ABCD_SIGPLAN_2000} that rely on building and solving simplified constraint systems (e.g., by considering only pair-wise inequalities) cannot discover such composite condition involving multiple variables. As a result, the SoftBound checks will remain in the $foo\_SB()$ and bound information of both pointers needs to be kept and propagated into $foo\_SB()$ at runtime, leading to high overhead.

In this paper, we show that rapidly removing all the bounds checking in \textit{foo}() is indeed possible using CHOP's statistical inference. Our solution stems from two key observations. First, redundant bounds checking can be effectively identified by comparing the value of $4*(tc1+tc2)+tc3+1$ with the value of $dsize$. In fact, all checks in $foo\_SB()$ can be eliminated if $4*(tc1+tc2)+tc3+1 \le dsize$. Next, through dependency analysis (detailed in section~\ref{subsec:sdDG}) along with profiling previous program executions, we find that the value of $4*(tc1+tc2)+tc3+1$ depends on the input arguments $snum$ and $ssize$ with positive coefficients, i.e., $4*(tc1+tc2)+tc3+1=3*snum+ssize+1$. Hence, given that $snum$, $ssize$ and $dsize$ values from past executions are safe, we can conclude that future executions are also guaranteed to be safe for any smaller values of $snum$ and/or $ssize$ and larger values of $dsize$. Combining the conditions derived from past executions, we can effectively derive a set of sufficient conditions (known as the safe region) for redundant check elimination. In general, CHOP will build a safe region with respect to the pointer-affecting variables based on all past executions, and update it as new data points become available. Future executions that satisfy the conditions of such safe region will be deemed as bound-safe. Note that it is possible that for some functions, we cannot infer the linear relationships among trip counts and function arguments. Hence we cannot perform function-level bounds check decision based on the function arguments, but have to get the values of pointer-affecting variables inside the function to bypass potential redundant bounds checking.

\section{System Design}
\label{sec:SD}
\begin{figure}[t]
	\centering 
	\includegraphics[scale=0.145]{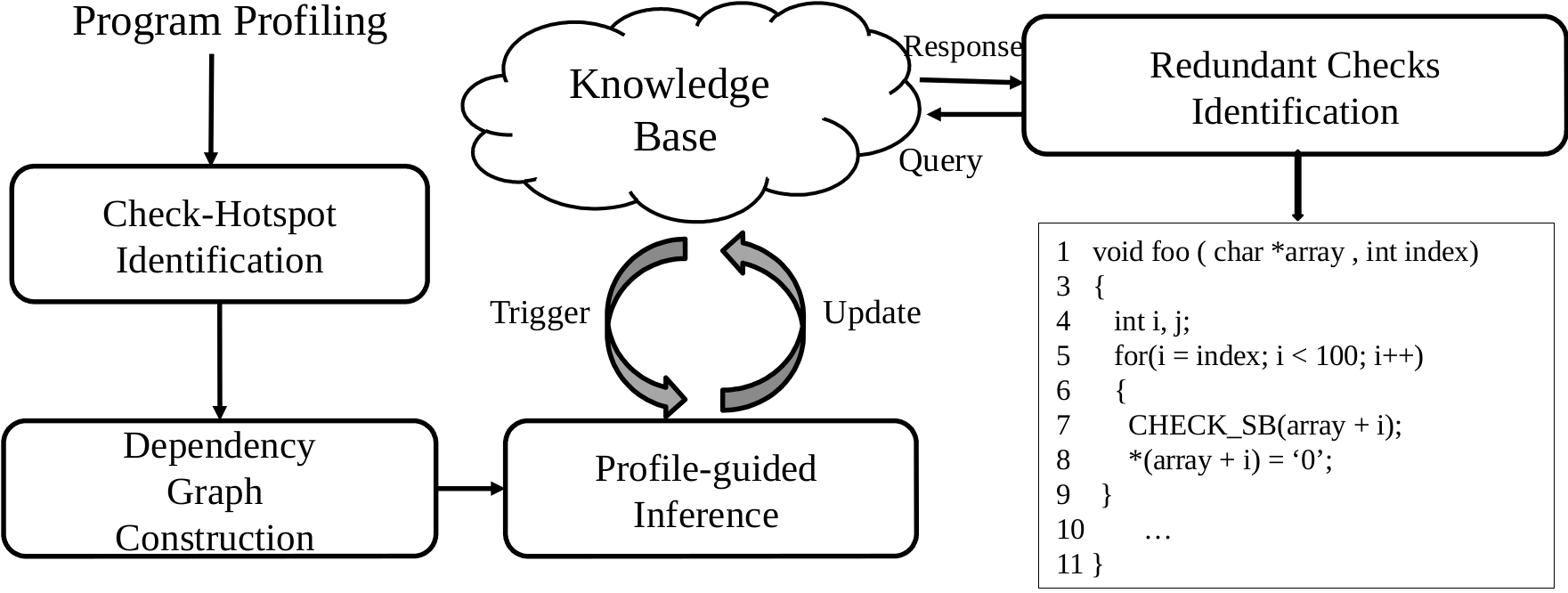}
	\caption{System Diagram}
	\label{fig:SD}
	\vspace{-.2in}
\end{figure}
CHOP consists of five modules: Dependency Graph construction, Profile-guided Inference, Knowledge Base, Runtime checks bypassing and Check-HotSpot Identification. Fig.~\ref{fig:SD} presents our system diagram. Given a target pointer, CHOP aims to determine if the pointer dereference needs to be checked. The {\bf pointer-affecting} variables, which can affect the value of target pointers (e.g., the base, offset and bound of the array). The rules for safe regions are then created based on the values of the pointer-affecting variables and are stored in the knowledge base as inferences for future executions. If the values of pointer-affecting variables satisfy the safe region rules, then the corresponding pointer dereference is considered to be safe.

\subsection{Dependency Graph Construction}
\label{subsec:sdDG}
Dependency Graph (DG) is a bi-directed graph $\mathcal{G}=(\mathcal{V}, \mathcal{E})$,  which represents program variables as vertices in $\mathcal{V}$ and models the dependency between the variables and pointers' bases/offsets/bounds through edges in $\mathcal{E}$. We construct a DG for each function including all of its pointers and the pointer-affecting variables.

\begin{definition}[DG-Node]
	The nodes in dependency graphs are pointers and the variables that can affect the value of pointers such as
	\begin{itemize}
		\item the variables that determine the base of pointers through pointer initialization, assignment or casting;
		\item variables that affect the offset and bound of pointers like array index, pointer increment and variables affecting memory allocation size;
		\item \textbf{Trip Counts} - the auxiliary variables to assist the analysis of loops. A trip count is the number of times a branch (in which a target pointer value changes) is taken.
	\end{itemize}
\end{definition}

\begin{definition}[DG-Edge]
	DG-Node $v_1$ will have an out-edge to DG-Node $v_2$ if $v_1$ can affect $v_2$.
\end{definition}

Algorithm~\ref{DG} shows the pseudo code of dependency graph construction for function $foo()$. First, we obtain all pointers and their pointer-affecting variables and represent them as DG-Nodes. Second, for each pair of identified DG-Nodes, we assign a  DG-Edges according to the rules in Remark~\ref{DGrules}.

\begin{algorithm}
	\caption{Dependency graph construction for a given function $foo()$} \label{DG}
	\small
	\begin{algorithmic}[1]
		\State \text{Input: source code of function $foo()$	} $  $
		\State $\text{Construct Abstract Syntax Tree, (AST) of function $foo()$}$  
		\State $\text{Initialize $\mathcal{V}=\phi$, $\mathcal{E}=\phi$}$
		\For {$\text{each variable $v$ in AST}$}	
		\State $\text{$\mathcal{V}= \mathcal{V} + \{v\}$}$
		\EndFor
		\For {$\text{each statement $s$ in AST}$}
		\For {$\text{each pair of variables $j,k$ in $s$}$}
		\State $\text{add edge $e(j,k)$ to  $\mathcal{E}$ according to Remark~\ref{DGrules}}$	   	
		\EndFor
		\EndFor
		\State \text{Output: Dependency-Graph $\mathcal{G}=(\mathcal{V},\mathcal{E})$	} $  $
	\end{algorithmic}
	
\end{algorithm}
\vspace{-.1in}

\begin{remark}
	\label{DGrules}
	\vspace{-.1in}
	\begin{framed}
		{\centering\bf \emph{Edges added into Dependency Graph:}}
		\begin{itemize}
			\item[]\textbf{E1}  Assignment statements\ \ \ \ \ \ \ \ \ \ A := $alpha\cdot$B \ \ \ \ \ \ \ \ ~\textit{B}$\,\to\,$\textit{A}
			\begin{itemize}[-]
				\item If constant $alpha$ is positive, then B is positively correlated to A
				\item If constant $alpha$ is negative, then B is negatively correlated to A
			\end{itemize}
			
			\item[]\textbf{E2} Function parameters\ \ \ \ \ \ \ \ \ \ Func(A,B)\ \ \ \ \ \ \ \ ~\textit{B}$\leftrightarrow$~\textit{A}
			
			\item[]\textbf{E3} Loops  \ \ for.../while...  \ ~\textit{Add Trip Counts to Loops}
			\item[](1) Assignment inside Loops\ \ \ \ \ \ \ \ A := B \ \ \ \ \ \ \ \ ~\textit{TC}$\,\to\,$\textit{A}
			\item[]\textbf{E4} Array Indexing\ \ \ \ \ \ \ \ \ \ \ \ \ \ \ \ \ \ \ \ \ \ A[i]\ \ \ \ \ \ \ \ \ \ \ \ \ ~\textit{i}$\,\to\,$\textit{A}
			
		\end{itemize}
		\label{fig:DG-rules}
		
	\end{framed}
	\vspace{-.1in}
\end{remark}

\subsection{Profile-guided Inference}
\label{subsec:sdSI}

Each function has its own dependency graph. After the dependency graph is built, it includes all pointers in the function and the pointer-affecting variables for all these pointers. We traverse dependency graph and identify adjacent DG-Nodes that represent the pointer-affecting variables associated with each target pointer. The target pointer will have an entry in the form of ${(func,ptr):(var_1,var_2,...,var_n)}$ where $func$ and $ptr$ stand for functions and pointers, with $var_i$ being the name of pointer-affecting variables associated with pointer $ptr$ in function $func$. By logging the values of these variables during program executions, we then build conditions for bypassing redundant runtime bounds check.

This module builds safe regions based on the pointer-affecting variables identified by dependency graphs and updates the safe regions through runtime data inference from previous execution. Once the pointer-affecting variables for the target pointer are identified as shown in Section \ref{subsec:sdDG}, CHOP will collect the values of pointer-affecting variables and produces a {\bf data point} in Euclidean space for each execution. The coordinates of each data point are the values of pointer-affecting variables. The dimension of the Euclidean space is the number of pointer-affecting variables for the target pointer. 

The inference about pointer safety can be derived as follows. Suppose a data point $p$ from prior execution with pointer-affecting variables $vp_1, vp_2, ..., vp_d$, is checked and deemed as safe. Another data point $q$ for the same target pointer but from another execution, is collected with pointer-affecting variables $vq_1, vq_2, ...,vq_d$. If each pointer-affecting variable of $q$ is not larger than that of $p$, e.g., $vq_1\le vp_1$, $vq_2\le vp_2$, ..., $vq_d\le vp_d$, then the bounds checking on the target pointer can be removed in the execution represented by $q$. Intuitively, if the increase of a variable value causes an increase of the index value or a decrease of the bound value, it will be denoted as positively correlated point-affecting variable. Similarly, the negatively correlated pointer-affecting variables are those cause decrease in index values (or increase in bound values) when they increase. The positively correlated pointer-affecting variables are safe when they are smaller and negatively correlated pointer-affecting variables are safe when they are larger. We unify the representations of pointer-affecting variables by converting a negatively-correlated variable $var_{neg}$ to $C-var_{neg}$ where C is a large constant that could be the maximum value of an unsigned 32-bit integer. Further, if multiple data points from prior executions are available, we integrate the safe conditions of individual data points to build a safe region for future inference.
%
%

As mentioned previously, the safe region is where pointer accesses are safe. In particular, the safe region of a single data point is the enclosed area by projecting it to each axis, which includes all input points that have smaller pointer-affecting variable values. For example, the safe region of a point $(3,2)$ is all points with the first coordinate smaller than 3 and the second coordinate smaller than 2 in $\mathbb{E}^2$. 

CHOP explores two approaches to obtain the safe region of multiple data points: \textit{union} and \textit{convex hull}. The union approach merges the safe regions generated by all existing data points to form a single safe region. We consider a larger safe region through building the convex hull of existing data points, and then deriving the linear condition of convex hull boundary as the condition for bypassing array bounds checking.

\subsubsection{Union}
Given a set $\mathcal{S}$ which consists of N data points in $\mathbb{E}^D$, where $D$ is the dimension of data points, we first project point $s_i \in \mathcal{S} , i=1,2,...,N$, to each axis and build N enclosed areas in $\mathbb{E}^D$, e.g., building safe region for each data point. The union of these N safe regions is the safe region of $\mathcal{S}$, denoted by $SR(\mathcal{S})$. {\it Thus, if a new data point $s_{new}$ falls inside $SR(\mathcal{S})$, we can find at least one existing point $s_k$ from $\mathcal{S}$ that dominates $s_{new}$. That is to say, the enclosed projection area of $s_k$ covers that of $s_{new}$, which means for every pointer-affecting variable, the $var_i$ of $s_k$ is larger than or equal to $var_i$ of $s_{new}$. Hence $s_{new}$ is guaranteed to be safe when accessing the memory.} Generally, when the index/offset variables of new data points are smaller than existing data points or the bound variable of new data point is larger than existing data point, the new data point will be determined as safe.

\subsubsection{Convex Hull}

Given a set of points $X$ in Euclidean space $\mathbb{E}^D$, convex hull is the minimal convex set of points that contains $X$, denoted by $Conv(X)$. In other words, convex hull of set $X$ is the set of all convex combination of points in $X$ as shown in equation \ref{eq:convexHull}. 
\begin{equation}
\label{eq:convexHull}
Conv(X)=\Big\{\sum_{i=1}^{|X|}\alpha_ix_i\Big| (\forall i:\alpha_i\ge 0)\wedge \sum_{i}^{|X|}\alpha_i=1\Big\}
\end{equation}

Based on prior $n$ execution samplings, the values of pointer-affecting variables are collected into a set $\mathcal{S}$ which consists of $n$ data points $\{s_i|i=1,2,...,n\}$ in $\mathbb{E}^D$. The convex hull of $\mathcal{S}$ is denoted by $CH(S)$. Suppose the number of pointer-affecting variables of target pointer is $D$, then each data point in $\mathcal{S}$ is a $D$ dimensional vector ($s_{i1}$,$s_{i2}$,...,$s_{iD}$). In this paper, CHOP employs the quickhull~\cite{quickhull} algorithm to construct the convex hull of all data points in $\mathcal{S}$ as the safe region. 

The reason of utilizing convex hull approach to construct safe region is that, it is bounded by the convex linear combination of data points, which is consistent with the linear constraints among pointer-affecting variables. If there exists a universal linear inequality of pointer-affecting variables for each $s_i$ with positive coefficients, then any point that is not outside the convex hull $CH(S)$ also has such linear inequality for its pointer-affecting variables, as stated in theorem ~\ref{theorem:convexhull}.

\begin{theorem}
	In $\mathcal{S}$, the coordinates of point $s_i$ is ($s_{i1}$,$s_{i2}$,...,$s_{iD}$). If $\forall i\in$ \{1,2, ...,$n$\}, $\sum_{j=1}^{D}\beta_j s_{ij}\le C$, then for all points $y_{m}\in CH(S)$, $\sum_{j=1}^{D} \beta_j y_{mj}\le C$ ($C$ is a constant and $\beta_j$ is the coefficient of $s_{ij}$).
	\label{theorem:convexhull}
\end{theorem}

\begin{proof}
	Given $y_{m}\in CH(S)$ and equation \ref{eq:convexHull}, we have 
	\begin{equation}
	y_{m} =\sum_{i=1}^{n}\alpha_is_i 
	\end{equation}
	where $\forall i: \alpha_i\ge 0$ and $\sum_{i}^{n}\alpha_i=1$. 
	
	Since $\forall i\in$ \{1,2, ...,$n$\}, $\sum_{j=1}^{D}\beta_j s_{ij}\le C$, then
	\begin{equation}
	\label{eq:sum}
	\alpha_i \sum_{j=1}^{D}\beta_j s_{ij}\le \alpha_iC
	\end{equation}
	
	By summing up equation \ref{eq:sum} for $i=1$ to $n$, we have 
	\begin{equation}
	\sum_{i=1}^{n}\alpha_i \sum_{j=1}^{D}\beta_j s_{ij}\le \sum_{i=1}^{n}\alpha_iC = C
	\label{eq:sumconv}
	\end{equation}
	
	Further convert equation \ref{eq:sumconv}:
	\begin{equation}
	\sum_{j=1}^{D}\beta_j\sum_{i=1}^{n}\alpha_i s_{ij}\le C
	\label{eq:sumconv2}
	\end{equation}
	
	Substitute $\sum_{i=1}^{n}\alpha_i s_{ij}$ by $y_{mj}$ in equation \ref{eq:sumconv2}, then 
	\begin{equation}
	\sum_{j=1}^{D} \beta_j y_{mj}\le C
	\end{equation}
	\label{proofconvex}
\end{proof}
Note that Theorem \ref{theorem:convexhull} can also be extended and applied to linear inequalities where the coefficients are not necessarily positive. As pointer bound information is added to dependency entry at the format of $(B-ptr\_bound)$, negatively related variable $var\_neg$ can be converted to new variables that have positive coefficients by $B-var\_neg$. 

\begin{equation}
\sum_{i=1}^{N}\beta_i \cdot VAR_i + \beta_{N+1}(B-P_{bound})<=C
\label{eq:generalRelation}
\end{equation}
Where $B$ is the bound of the target pointer and $C$ is a large constant such as $2^{32}-1$.
Equation \ref{eq:generalRelation} represents a hyperplane in $\mathbb{E}^{N+1}$ which separates the $\mathbb{E}^{N+1}$ space into two $\mathbb{E}^{N+1}$ subspaces. All normal data points that have legitimate pointer operations will fall inside the same subspace. Hence the convex hull built from these normal data points will be contained in this subspace which means all bounds checking on points falling inside or on the facet of this convex hull are safe to be removed.

Thus, if the new data point is inside corresponding convex hull, then the check can be removed and this check bypassing is guaranteed to be safe. 

Convex hulls with low dimensions are easy to represent. Before convex hull construction, we will use dimension-reduction techniques like PCA to eliminate dimensions that have internal linear relationship. This is equivalent to filtering out the planar points in a lower dimensional hyperplane from the convex hull. If the convex hull turns out to be one dimensional (line) or two dimensional (plane), then we can easily represent it as an inequality. For higher-dimensional convex hull, we will store the convex hull and verify the check bypassing condition by determining if the point lies inside, outside the convex hull or on its facet, which will be described in the section\ref{subsec:sdKB}. Compared with Union approach, the safe region built by Convex Hull approach is expanded and can achieve higher redundant checks bypassing under the assumption that pointer-affecting variables are linearly related to target pointers.

\begin{figure}[th]
	\noindent 
	\includegraphics[scale=0.4]{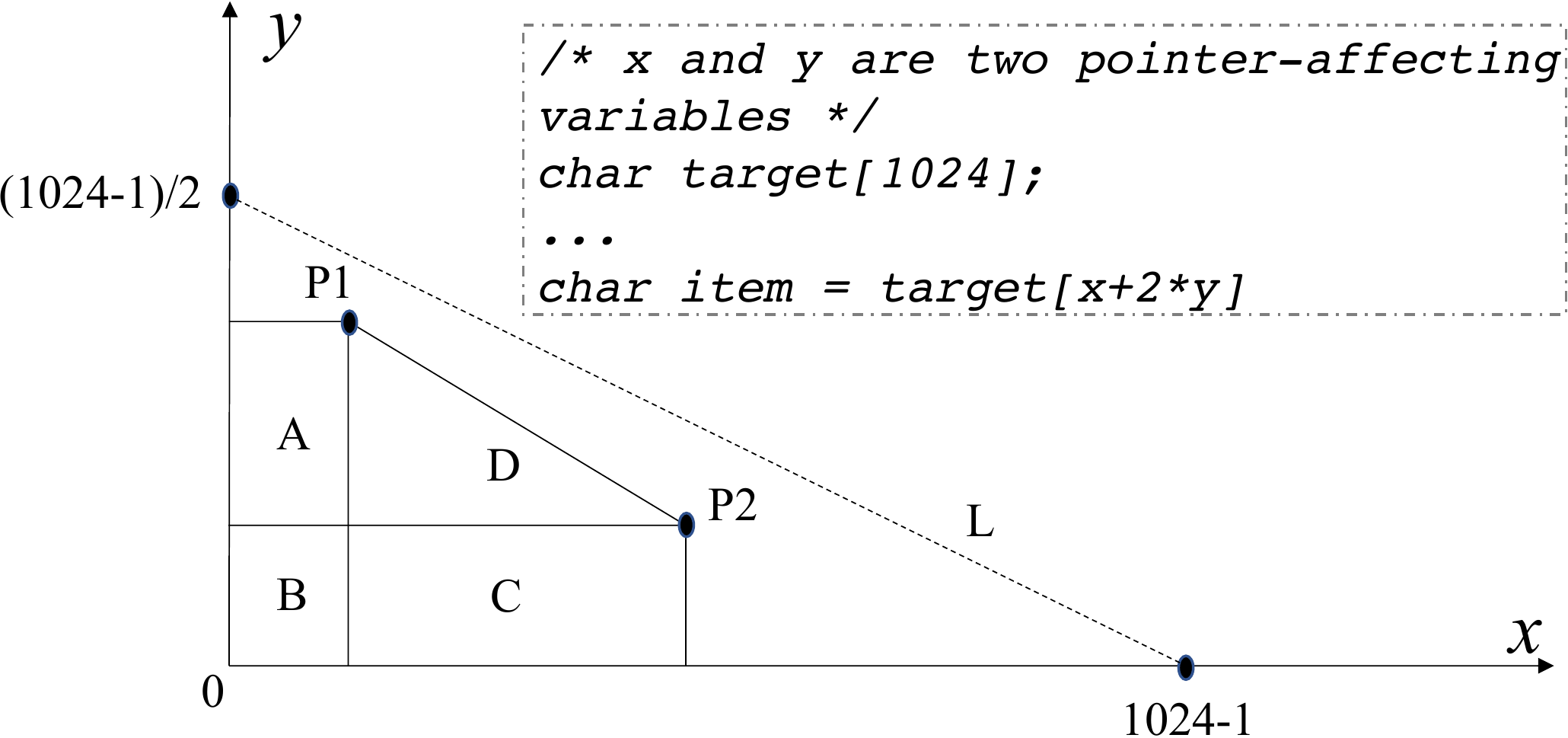}
	\caption{Convex Hull vs Union in a two-dimensional example}
	\label{fig:convexhull2d}
	\vspace{-.1in}
\end{figure}

We further illustrate the performance gain of convex hull approach over union approach using the example shown in Fig~\ref{fig:convexhull2d}. We have a two-dimensional case where the array $target$ has two pointer-affecting variables $x$ and $y$ (the bound of $target$ in this example is a fixed value of 1024 and the maximum value of legitimate index is 1023). $target$ will be dereferenced at the index $x+2*y$ and the value is passed to $item$. Suppose during one sampling execution, we observed that the value of $x$ is 160 and the value of $y$ is 400. So we have one data point, represented by $P1:$ $P1=(160,400)$. Similarly, from another execution, we get $P2=(510,170)$. And both of $P1$ and $p2$ are deemed as safe because SoftBound has performed bounds checking in those two executions. When the safe region is built by the union approach, it will be the union of area $A$, $B$ and $C$. On the other hand, if the safe region is built by the convex hull approach, then it will be the union of the area $A$,$B$,$C$ and $D$. The optimal condition of the safe region by convex hull would be the area enclosed by $x$ axis, $y$ axis and the dashed line $L$. It also shows that the safe region CHOP derives will never go beyond the theoretically optimal safe region. The reason is that if any new data point falls within the top-right side of $L$, CHOP can tell that it's outside the current safe region. Hence, the default SoftBound checks will be performed and out-of-bound access will be detected. Such data points will not be collected by CHOP for updating the safe region, i.e., all data points that CHOP collects will be within the theoretically optimal safe region. Therefore, the convex hull built by these data points will be a subset of the optimal safe region.


\subsubsection{Safe Region Update}

There are data points that can not be determined as safe or not by current safe region but later verified as legitimate. Such data points can be used to dynamically update the safe region. Given current safe region $SR(\mathcal{S})$ and the new coming data point $s_{new}$, $SR(\mathcal{S})$ will be updated to $SR(\mathcal{S})'$ by:
\begin{equation}
\label{eq:srUpdate}
SR' = SR(\mathcal{S} \cup s_{new}) = SR(\mathcal{S}) \cup SR(s_{new}) = SR(\mathcal{S}) \cup \mathcal{T},
\end{equation}
where $\mathcal{T}$ is the set of safe points inside $SR(s_{new})$ but outside $SR(\mathcal{S})$. If $\mathcal{T}$ is empty which means $SR(s_{new})$ is contained by $SR(\mathcal{S})$, then there is no need to update the safe region $SR(\mathcal{S})$. Otherwise the update of safe region encapsulates two scenarios: 
\begin{itemize}
	\item There are positively correlated pointer-affecting variables (such as  array index) of $s_{new}$ that have larger values than corresponding pointer-affecting variables of all points in $SR(S)$, 
	\item There are negatively correlated pointer-affecting variables (such as  bound of pointers) of $s_{new}$ that are smaller than those of all points in $SR(S)$
\end{itemize}

When one or both of above scenarios occur, the safe region will be enlarged to provide a higher percentage of redundant bounds check bypassing. The safe region is updated in a different thread from that of bounds checking so that it will not contribute to the execution time overhead of the tesing program. 

\subsection{Knowledge Base}
\label{subsec:sdKB}
\begin{algorithm}[h]
	\caption{Algorithm for deciding if a point is in convex hull: $isInHull()$} 
	\label{alg:inhull}
	\small
	\begin{algorithmic}[1]
		\State \text{Input: convex hull represented by $m$ facets $F = {f_1,f_2...,f_m}$}
		\State \text{Input: normal vectors (pointing inward) of facets $N = {n_1,...n_m}$}
		\State \text{Input: new data point $p$}
		\State \text{Output: isInHull}
		\State \text{Init: isInHull = True}
		\For {$\text{i in [1,m]}$}
		\State \text{$cur\_facet = f_i$}
		\State \text{randomly select a point $p_i$ from $f_i$}
		\State \text{$v_i = p-p_i$}
		\If {$v_i \cdot n_i$ is negative}
		\State \text{isInHull = False}
		\EndIf
		\EndFor
	\end{algorithmic}
\end{algorithm}
CHOP stores the safe regions for target pointers in a disjoint memory space - the Knowledge Base. The data in Knowledge Base represents the position and the sufficient conditions for bypassing the redundant bounds checking for each target pointer. Runtime Profile-guided Inference can be {\it triggered} to compute the Safe Region by Knowledge Base when we detect redundant checks, then the Knowledge Base can be {\it updated} as more execution logs are available.
\subsubsection{Union}

For the \textit{Union} approach, the values of pointer-affecting variables of the target pointer are kept in the knowledge base. Suppose we have a number of $K$ prior executions associated with pointer $p$ which has $D$ pointer-affecting variables, then a $K$x$D$ matrix $U_{KxD}$is stored. When performing the bounds checking for pointer $p$ in new executions, the value of $p$'s pointer-affecting variables are compared with those stored in $U$. Once a row in $U$ can dominates $p$'s pointer-affecting variables, which means the new data points that represents this new execution is inside one existing point, $p$ is considered as safe in this new execution.

\subsubsection{Convex Hull}

CHOP determines whether the new point is in Safe Region by the following method. For each facet of the convex hull, the normal vector of the facet is also kept in Knowledge Base. For a convex hull in $\mathbb{E}^D$, $D$ points that are not in the same hyperplane are sufficient to represent a $D$ dimensional facet. Suppose the convex hull built from sampling runs has $M$ facets. These $M$ data points are stored in the Knowledge Base as a hashmap in the format of $(d : data_d)$, where $d$ is the ID/index of data point and $data_d$ is the coordinates of data point $d$. Then this convex hull can be represented as a $M$x$D$ matrix $CH_{MxD}$ where each element is the index of sampling points and each row represents a $D$ dimensional facet. Now a new data point $T$ is available and CHOP will decide whether this new point is inside or outside each of the $M$ facets. Let $N_i$ be the normal vector (pointing inwards) of each facet $f_i$ , $\forall i=1,2,...,M$. From facet $f_i$, randomly choose one point denoted by $CH[i][j]$, link $CH[i][j]$ and the new point $T$ to build a vector $\vec{V_i}$. If the inner product $P_i$ of $\vec{V_i}$ and $N_i$ is positive which means the projection of $\vec{V_i}$ to $N_i$ has the same direction with $N_i$, then point $T$ is in the inner side of facet $f_i$. Repeat this for each facet, eventually, if $P_i \ge 0 \forall$ i=1,2,...,M, then the new point $T$ is inside the convex hull or right on the facets. We embed this process into function $isInHull()$ and demonstrate it in Algorithm~\ref{alg:inhull}.


If the data points are of one dimension, we store a threshold of pointer-affecting variable as the safe region for checks elimination. For higher dimensional data points, in case the safe region becomes too complex, we can store a pareto optimal safe region of less data points instead of the union of safe regions.


\subsection{Bypassing Redundant Array bounds checking}
\label{subsec:sdRuntime}

We instrument source code of benchmark programs to add a {\it CHECK\_CHOP}() function. {\it CHECK\_CHOP}() verifies the condition of bounds check elimination by comparing pointer-affecting variables collected from new executions with statistics from knowledge base. 

Two levels of granularity for redundant bounds check bypassing are studied: function level and loop level. 
\begin{inparaenum}[a)]
	\item Function-level redundant bounds check bypassing conditions are verified before function calls. If the new data point is inside the built safe region, the propagation of bound information and the bounds checking can be removed for the entire function.
	\item Memory access through pointers inside loops are most likely responsible for the high overhead of SoftBound checks. Loop-level redundant bounds check bypassing is performed when the condition doesn't hold for all target pointer dereferences inside the function. Instead of bypassing all bounds checking for target pointer in the function, the condition for bypassing bounds checking inside loops are examined. We ``hoist'' the bounds checking outside the loop. The safe region check is performed before the loop iterations. If the pointer accesses in one iteration are considered safe, then the bounds checking inside this iteration can be skipped.
\end{inparaenum}
%

\subsection{Check-HotSpot Identification}
\label{subsec:sdHot}


In order to maximize the benefit of our runtime check bypassing while maintaining simplicity, CHOP focuses on program {\em Check-HotSpots}, which are functions associated with intensive pointer activities and resulting in highest overhead to bounds checking. 



CHOP identifies Check-HotSpots using three steps as follows:
\begin{inparaenum}[a)]
	\item Profiling testing program: We use Perf profiling tool~\cite{perf} to profile both a non-instrumented program and its SoftBound-instrumented version. The execution time of each function (with and without bounds checking) are recorded.
	\item Analyzing function-level overhead $O_f$: For each function $f$, we calculate its individual contribution to bounds check overhead, which is the time spent in $f$ on SoftBound-induced bounds checking, normalized by the total overhead of the testing program. More precisely, let $T$ and $\hat{T}$ be the execution time of the non-instrumented, target program and its SoftBound-instrumented version, and $t_f$ and $\hat{t}_f$ be the time spent in function $f$, respectively. We have $O_f = (t_f - \hat{t}_f )/(T-\hat{T})$. 	
	\item Identifying Check-HotSpots: In general, we select all functions with at least 5\%\footnote{We use 5\% as threshold to identify Check-HotSpots, while this number can be varied depending on users' preference.} function-level overhead as the Check-HotSpots, which will be the target functions for bounds check bypassing.
	
\end{inparaenum}

In our evaluation, we consider two different types of applications: interactive applications and non-interactive applications. For non-interactive applications, such as SPEC2006 benchmark, we use the testing inputs provide with the benchmark. For interactive applications (such as web severs and browsers) that require user inputs and exhibit diversified execution paths, we intentionally generate user inputs/requests that are able to trigger desired execution paths of interest (e.g., containing potential memory bounds check vulnerabilities). Check-HotSpots are then identified accordingly.

\subsection{Example}
\label{sec:SDExample}

CHOP instruments the code by adding two new branches as shown in Fig.\ref{fig:CHOPCode}. The function {\it CHECK\_CHOP()} verifies if the inputs to function {\em foo()} satisfy the conditions for bounds check bypassing(i.e., in the safe region). Then, one of the branches is selected based on whether bounds checking are necessary.

Recall the SoftBound instrumented {\it foo\_SB}() function from Fig. \ref{fig:SBCode} . We add trip counts $tc1$, $tc2$ and $tc3$ for the three cases in lines 17, 24 and 31, respectively. According to CHOP's dependency graph construction, there are edges from node $tc1$, $tc2$ and $tc3$ to pointer node \textit{cp2}. Further, the values of $4*(tc1+tc2)+tc3$ is determined by input variable $ssize$ and $snum$, producing edges from \textit{ssize} and \textit{snum} to trip counts in the dependency graph. Thus, the pointer-affecting variables for pointer {\it cp2} are $(ssize,snum,C-dsize)$. Suppose that constant $C$ is defined as $2^{32}-1$ (i.e., the maximum 32-bit unsigned integer), and that we have three past executions with pointer-affecting variable values as follows: $p_1=(200,60,2^{32}-1-256)$, $p_2=(180,20,2^{32}-1-256)$ and $p_3=(150,40,2^{32}-1-512)$. The safe region for check elimination will be built based on above three data points $p_1,p_2,p_3$ in a $\mathbb{E}^3$ space according to the approach described in section \ref{subsec:sdSI}. 

In future executions, any input to function {\em foo()} generates a new data point $p$ with pointer-affecting variables $(p_{ssize}, p_{snum}, p_{dsize})$ for examination. It is verified by {\it CHECK\_CHOP}() to determine if point $p$ is inside this safe region, in order to decide whether bounds check elimination is possible. In particular, in the union approach, as long as we can find one existing point $p_i$ (from $p_1,p_2,p_3$) that Pareto-dominates $p$, i.e., any    pointer-affecting variables (i.e., components) of $p_i$ is greater than or equal to that of $p$, 
then the memory access of pointer $foo\_SB:cp2$ is determined to be safe. 
In convex hull approach, we need to solve the convex hull containing points $p_1,p_2,p_3$. With sufficient data, the boundary of constructed convex-hull safe region can be derived as $ssize+3*snum+1\le dsize$, e.g., after all corner points have been observed. Similar to~\cite{ABCD_SIGPLAN_2000}, CHOP applies to both base and bound of arrays/pointers.

\section{Implementation}
\label{sec:implementation}

In this section, we explain in details how our system is implemented. Fig.~\ref{fig:Implemenation} shows the framework of CHOP with the tools and algorithms we use in different modules. Modules in CHOP can be categorized into two components: Pre-Runtime Analysis and Runtime Checks Bypassing. Pre-Runtime Analysis can be executed offline and Runtime Checks Bypassing is used to identify and avoid redundant checks during runtime. To obtain Check-HotSpot, we profile non-instrumented programs as well as SoftBound-instrumented programs to get the differences of execution time of user-level functions. Based on Check-HotSpot results, we used Static Code Analysis tool Joern to construct and traverse Dependency Graph to get pointer-affecting variables for target pointers. By logging the sampling executions, we build the union or convex hull for safe regions as the check bypassing conditions and store them in database for future inferences.
\begin{figure}
	\hspace{0.4in}
	\includegraphics[scale=0.5]{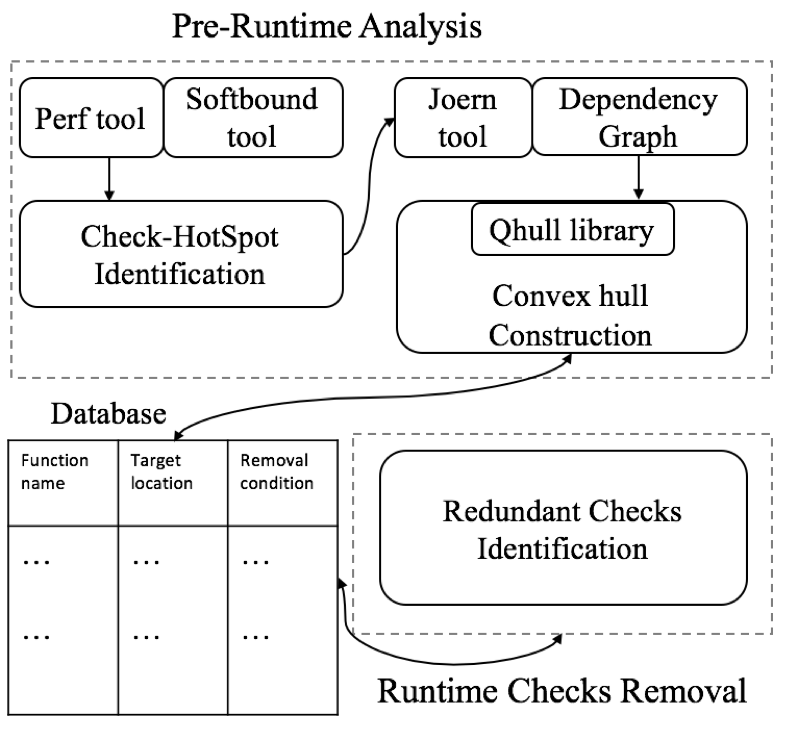}
	\caption{System framework and Implementation in Pre-Runtime Analysis and Runtime Checks Bypassing modules with the tools used in different components}
	\label{fig:Implemenation}
\end{figure}

\begin{table}[]
	\centering
	\scriptsize
	\begin{tabular}{c|c|c|c|}
		\cline{2-4}
		& \multicolumn{3}{c|}{Time spent in}                                                                                                                                                               \\ \hline
		\multicolumn{1}{|c|}{Function}    & \begin{tabular}[c]{@{}c@{}}Non-instrumented \\ version (s)\end{tabular} & \begin{tabular}[c]{@{}c@{}}SoftBound-instrumented\\  version(s)\end{tabular} & \begin{tabular}[c]{@{}c@{}}SoftBound overhead\\  breakdown\end{tabular} \\ \hline
		\multicolumn{1}{|c|}{\textit{F1}} & 814.87                                                                  & 1580.03                                                                      & 32.08\%                                                                 \\ \hline
		\multicolumn{1}{|c|}{\textit{F2}} & 977.08                                                                  & 1582.68                                                                     & 25.39\%                                                                 \\ \hline
		\multicolumn{1}{|c|}{\textit{F3}} & 1.67                                                                    & 353.76                                                                        & 14.76\%                                                                 \\ \hline
		\multicolumn{1}{|c|}{\textit{F4}} & 148.85                                                                  & 270.84                                                                    & 5.11\%                                                                  \\ \hline
	\end{tabular}
	
	\caption{Check-Hospot functions out of total 106 called-functions from sphinx. Function names from F1 to F4, F1: vector\_gautbl\_eval\_logs3, F2: mgau\_eval, F3: utt\_decode\_block and F4: subvq\_mgau\_shortlist. }
	\label{Hospots}
	
\end{table}

\subsection{Dependency Graph Construction}

Plenty of static code analysis tools exist for parsing and analyzing source code, such as~\cite{Pscan}~\cite{splint}~\cite{PR-Miner}~\cite{srivastava2011security}. In this paper, we analyze the code in the format of Abstract Syntax Tree (AST)~\cite{AST_ACSA_2012} to build the dependency graph such that we can easily obtain variable declaration types, statement types and dependencies of variables. An AST reveals the lexical structure of program code in a tree with variables/constants being the leaves and operators/statements being inner nodes. We build the AST using Joern~\cite{joern}. Joern is a platform for static analysis of C/C++ code. It generates code property graphs (such as CFG and AST) and stores them in a Neo4J graph database. 

Among all the variables in a function, we are only interested in the pointers and pointer-affecting variables, e.g., a sub-AST for each function without other redundancy. For this purpose, we instrument AST API from Joern with the idea from~\cite{AST_ACSA_2012}, where extracting sub-AST on the function level is studied.

The rules of constructing DG are shown in Algorithm~\ref{DG} and Remark~\ref{DGrules} in section~\ref{sec:SD}. After the dependency graph is constructed, we traverse the dependency graph to identify pointers for bounds check bypassing and use light-weight tainting~\cite{yamaguchi2014modeling,yamaguchi2015automatic} to determine the set of nodes connected to target pointers from the dependency graph.

\subsection{Statistical Inference and Knowledge Base}
\label{sec:sikb}
We employ Quickhull algorithm (Qhull)~\cite{quickhull} to compute the convex hull. For a target pointer that has $D$ pointer-affecting variables (including ($C-ptr\_bound$)) and $n$ sampling runs, we first generate $n$ points in $\mathbb{E}^D$, then select the $D$ boundary points w.r.t each axis. As a result, a total of n+D points in $\mathbb{E}^D$ will be the input of convex hull construction. In the running example, given the prior statistics mentioned in Section \ref{sec:SDExample}, these 6 points are \begin{inparaenum}[a)]
	\item $(200,60,2^{32}-257)$,
	\item $(180,20,2^{32}-257)$,
	\item $(150,40,2^{32}-513)$,
	\item $(200,0,0)$,
	\item $(0,60,0)$,
	\item $(0,0,2^{32}-257)$.
\end{inparaenum}

Note that integer overflow is a special case in bounds checking. Assume an array $arr$ is accessed by using $arr[x + 1]$. A data point that was collected is ($x = UINT\_MAX$). Since the expression $x1+1$ overflows, $arr$ is accessed at position 0, which is safe if the array contains at least one element. By the default rules in CHOP, ($UINT\_MAX-1$) would then be determined to be inside a safe region. However, since no integer overflow occurs, $arr$ at position $UINT\_MAX$ is accessed, which would result in an out-of-bound access. Hence, CHOP performs special handling of integer overflow when updating the safe region. Suppose the data point that causes an integer overflow ($x = UNIT\_MAX$ in this case) is observed, when it is used to update the convex hull, we will calculate if there is an integer overflow, by determine if $UNIT\_MAX \ge UNIT_MAX - 1$. In general, if $arr$ is accessed by $arr[x+y]$, we will check if $x \ge UNIT\_MAX-y$ when updating the convex hull. If it holds, then we discard this data point and do not update the convex hull. Since convex hull updating happens offline, it will not affect the runtime overhead of bounds checking.

We use SQLite~\cite{owens2010sqlite} to store our Knowledge Base. We created a table, which has fields including function names, pointer names and the corresponding conditions for redundant checks bypassing(e.g., the matrix mentioned in Section \ref{subsec:sdKB}).

For Union safe region, if the data points are of one dimension, we store a threshold of pointer-affecting variable as the safe region for checks elimination. If the data points are of higher dimension, we only store the data points that are in the boundary of the union area. For Convex hull approach, we store the linear condition of the safe region boundary in the case of low-dimensional data points and a set of frontier facets (as well as their normal vectors) in the case if high-dimensional data points.

\subsection{Bypassing Redundant Checks}

For function-level redundant bounds check bypassing, we maintain two versions of Check-Hotspot functions: the original version (which contains no bounds checking) and the SoftBound-instrumented version that has bounds checking. By choosing one of the two versions of the function to be executed based on the result of {\it CHECK\_CHOP}() verification, we can either skip all bounds checking inside the function (if the condition holds) or proceed to call the original function (if the condition is not satisfied) where bounds checking would be performed as illustrated in Fig. \ref{fig:CHOPCode}.  The instrumentation of {\it CHECK\_CHOP()} condition verification inside functions leads to a small increase in code size (by about 2.1\%), and we note that such extra code is added only to a small subset of functions with highest runtime overhead for SoftBound (see Section~\ref{subsec:sdHot} for details). 

While function-level redundant bounds check bypassing applies to the cases where all the target pointer dereferences are safe inside the target function, loop-level removal provides a solution for pointer dereferences that can only be considered as partially safe when memory accesses inside loops are closely related to the loop iterator. The safety of pointer dereferences can be guaranteed when the value of loop iterator falls within certain range. In this case, we consider the loop iterator as a pointer-affecting variable and incorporate iteration information into the safe region. We duplicate loops similar to duplicating functions. Before entering the loop, the function \textit{CHECK\_CHOP()} is called and if all bounds checking inside the loop are considered safe, the check-free version of the loop will be called.

\subsection{Check-Hotspot Identification}
The program profile tool Perf is used to identify the Check-Hotspot functions. Perf is an integrated performance analysis tool on Linux kernel that can trace both software events (such as page faults) and hardware events (such CPU cycles). We use Perf to record the runtime overhead of target functions. 

We compile our test cases with SoftBound. For each Check-Hotpost function, we calculate the time spent in non-instrumented version and SoftBound-instrumented version, then calculate the difference between them to get the overhead of SoftBound checks. After ranking all functions according to the execution time overhead of SoftBound checks, we consider functions that contributes over 5\% bounds checking overhead as Check-Hotspot functions. Noting that we pick 5\% threshold for Check-Hotspot in this paper, but it can be customized depending on specific usages. 

TABLE~\ref{Hospots} shows the results for Check-Hotspot Identification for the application \textit{Sphinx3} from SPEC. In total, the four functions listed in the table contribute over 72\% runtime overhead of SoftBound. 

\section{Evaluation}
\label{sec:eva}

We use SoftBound (version 3.4) as the baseline to evaluate the effectiveness of CHOP on bypassing redundant bounds checking. All experiments are conducted on a 2.54 GHz Intel Xeon(R) CPU E5540 8-core server with 12 GByte of main memory. The operating system is Ubuntu 14.04 LTS.
We select two different sets of real world applications: (i) Non-interactive applications including 8 applications from SPEC2006 Benchmark suite, i.e., \textit{bzip2}, \textit{hmmer}, \textit{lbm}, \textit{sphinx3}, \textit{libuquantum}, \textit{milc}, \textit{mcf} and \textit{h264ref}; 
and (ii) Interactive applications including a light-weight web server thttpd (version beta 2.23). In the evaluation, we first instrument the applications using SoftBound and employ Perf to identify the Check-HotSpot functions in all applications. Similar to ABCD~\cite{ABCD_SIGPLAN_2000}, we consider the optimization of upper- and lower-bounds checking as two separated problems. In the following, we focus on redundant upper-bounds checking while the dual problem of lower-bounds checking can be readily solved with the same approach. In the experiments, we use both Union and Convex Hull approaches to obtain bounds check bypassing conditions for Check-HotSpot functions. We further compare the performance between these two approaches if they have different conditions for bounds check decisions. The inputs we used for testing SPEC applications are from the \emph{reference} workload provided with SPEC benchmarks. For thttpd, we created a script that randomly generate urls with variable lengths and characters, then send them together with thttpd requests to the server for evaluation. In general, for applications that do not provide developer supplied representative test cases, we note that fuzzing techniques~\cite{mcnally2012fuzzing}~\cite{woo2013scheduling} can be applied to generate test cases. The policies considered in our evaluation are \begin{inparaenum}[a)]
	\item SoftBound instrumentation (denoted as \textbf{SoftBound}).
	\item  CHOP Optimized SoftBound with redundant bounds check bypassing (denoted as \textbf{C.O.S}). 
\end{inparaenum}

Our Check-HotSpot identification identifies 23 functions from all 9 applications mentioned above. For example, in the application \textit{bzip2}, the bounds check overhead of the three functions \textit{bzip2::mainGtU}, \textit{bzip2::generateMTFValues} and \textit{bzip2::BZ2\_decompress} contribute 68.35\% to the total bounds check overhead in \textit{bzip2}. Similarly, we studied 98.01\% bounds check overhead in \textit{hmmer}, 86.19\% in \textit{lbm}, 62.58\% in \textit{sphinx3}, 72.71\% in \textit{milc}, 94.18\% in \textit{libquantum}, 69.55\% in \textit{h264ref}, 69.51\% in \textit{mcf}, and 83.56\% in \textit{thttpd}. We note that some Check-HotSpot functions contribute much more than others to SoftBound overhead mainly because they are executed frequently, e.g., \textit{bzip2::mainGtU} is called more than 8 million times, even though they have small code footprints.


\subsection{Removal of Redundant Array bounds checking}
\begin{figure}[t]
	\begin{center}
		\centering	
		\includegraphics[scale=0.145]{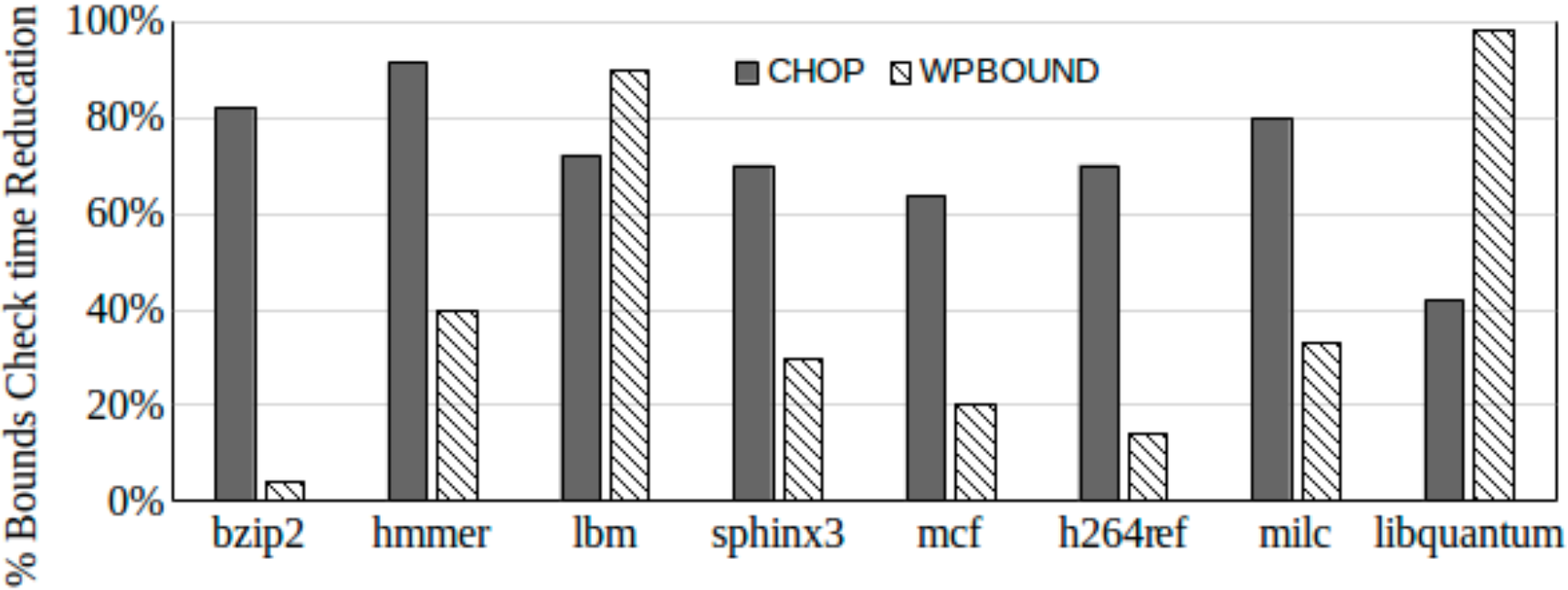}
		\caption{Comparison of normalized execution time overhead reduction between C.O.S and WPBound}
		\label{fig:cmp-runtime}
	\end{center}
	\vspace{-.2in}
\end{figure}

\begin{figure}[t]
	\begin{center}
		\centering	
		\includegraphics[scale=0.145]{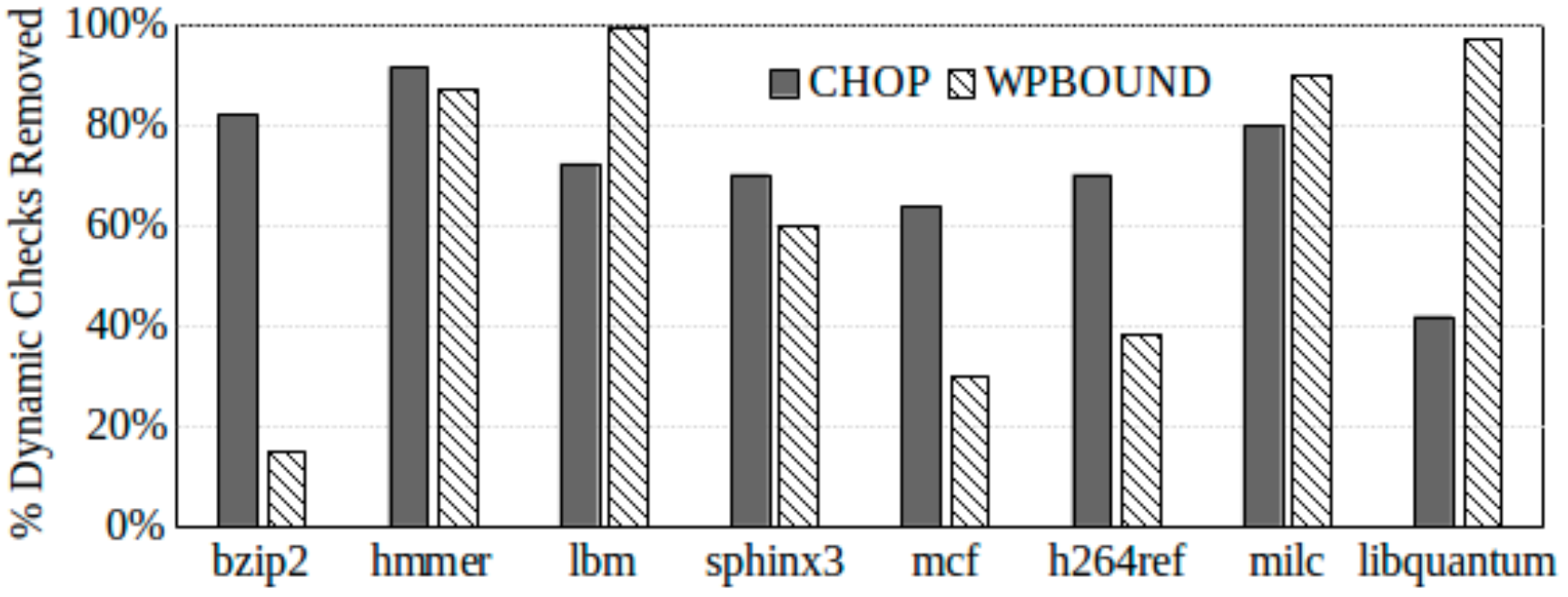}
		\caption{Comparison of bounds check removal ratio between C.O.S and WPBound}
		\label{fig:cmp-ratio}
	\end{center}
	\vspace{-.3in}
\end{figure}
Fig.~\ref{fig:cmp-runtime} shows the comparison of execution time overhead reduction over SoftBound between C.O.S and WPBound, normalized by the execution time of original applications without performing bounds checking. In particular, we measure the runtime overhead for each Check-HotSpot functions, before and after enabling CHOP. Due to the ability to bypass redundant bounds checking, C.O.S. achieves significant overhead reduction. The highest reduction achieved by CHOP is~\textit{hmmer}, with a 86.94\% execution time reduction compared to SoftBound. For~\textit{bzip2}, ~\textit{lbm},~\textit{sphinx3} and~\textit{thttpd}, SoftBound overheads are decreased from 39\% to 8\% , 55\% to 18\%, 31\% to 11\% and 66\% to 12\%. Overall, CHOP achieved an average execution time reduction of 66.31\% for Check-HotSpot functions, while WPBound achieves 37\% execution overhead reduction on SPEC benchmarks.

To illustrate CHOP's efficiency in bypassing redundant bounds checking, Table~\ref{Checks} and Table ~\ref{Comparision} show the number of total bounds checking required by SoftBound and the number of redundant checks bypassed by CHOP, along with rate of false positives reported in C.O.S. Basically, we observed no false positives in our experiments. In particular, Table~\ref{Comparision} compares the performance of the Union and Convex Hull approaches in terms of the number of runtime bounds checking being bypassed in two Check-HotSpot functions from thttpd. Since both Check-HotSpot functions have high-dimensional bound affecting variables, the conditions derived using Union and Convex Hull are different. More precisely, the safe regions constructed by Convex Hull dominates that of the Union approach, leading to higher bypassing ratio. CHOP successfully removes 82.12\% redundant bounds checking through Convex Hull approach comparing to 59.53\% by Union approach in \textit{thttpd::defang}. Also, we compare the number of bounds checking removed between CHOP and WPBound as shown in Fig.~\ref{fig:cmp-ratio}. On average, CHOP successfully bypassed 71.29\% runtime bounds checking for SPEC benchmarks while WPBound can achieve 64\%. We noted that for some of the testing programs such as {\it lbm, milc and libquantum}, WPBound outperforms our Convex Hull approaches. We examined the reasons and explained it in section~\ref{sec:case-others}.


\begin{table*}[h]
	\footnotesize
	\centering
	\bgroup
	\def\arraystretch{0.9}%
	\resizebox{1.0\textwidth}{!}{%
		\begin{tabular}{|c|c|c|c|c|c|}
			\hline
			
			\multicolumn{1}{|c|}{\textbf{Benchmark::Function Name}}          & \textbf{Total bounds checking}             & \textbf{Redundant checks bypassed}        & \textbf{False Positive}       \\ \hline
			\multicolumn{1}{|c|}{\textit{bzip2::generateMTFValues}} & 2,928,640                   & 1,440,891 (49.2\%)                      & 0 (0.0\%)    \\ \hline
			\multicolumn{1}{|c|}{\textit{bzip2::mainGtU}}     & 81,143,646                   & 81,136,304 (99.9\%)                     & 0 (0.0\%)         \\ \hline
			
			\multicolumn{1}{|c|}{\textit{bzip2::BZ2\_decompress}}      & 265,215                   & 196,259 (74.0\%)                      & 0 (0.0\%)           \\ \hline
			\multicolumn{1}{|c|}{\textit{hmmer::P7Viterbi}}      & 176,000,379                  & 124,960,267 (71.0\%)                       & 0 (0.0\%)        \\ \hline
			\multicolumn{1}{|c|}{\textit{lbm::LBM\_performStreamCollide}}      & 128277886                  & 128277886 (100.0\%)                      & 0 (0.0\%)        \\ \hline
			\multicolumn{1}{|c|}{\textit{sphinx3::vector\_gautbl\_eval\_logs3}}      & 2,779,295                  & 2,779,295 (100.0\%)                     & 0 (0.0\%)        \\ \hline
			\multicolumn{1}{|c|}{\textit{sphinx3::mgau\_eval}}      & 725,899,332                  & 725,899,332 (100.0\%)                      & 0 (0.0\%)        \\ \hline
			\multicolumn{1}{|c|}{\textit{sphinx3::subvq\_mgau\_shortlist}}      & 24,704                 &  4,471 (18.1\%)                      & 0 (0.0\%)        \\ \hline
			\multicolumn{1}{|c|}{\textit{thttpd::httpd\_parse\_request}}      & 9,990                 &  9,990  (100.0\%)                      & 0 (0.0\%)        \\ \hline
			\multicolumn{1}{|c|}{\textit{thttpd::handle\_newconnect}}      & 9,121                 &  7,300  (80.0\%)                      & 0 (0.0\%)        \\ \hline
		\end{tabular}		
	}
	\egroup
	\caption{ Number of bounds checking required by SoftBound and bypassed by CHOP in each Check-HotSpot function through Convex Hull Optimization}
	\label{Checks}
	\hspace{1mm}	
\end{table*}

\begin{table*}[h]
	\footnotesize
	\centering
	\bgroup
	\def\arraystretch{0.9}%
	\resizebox{1.0\textwidth}{!}{%
		\begin{tabular}{|c|c|c|c|c|}
			\hline
			\multirow{2}{*}{\textbf{Benchmark::Function name}} & \multirow{2}{*}{\textbf{Total bounds checking}} & \multicolumn{2}{c|}{\textbf{Redundant checks bypassed}} & \multirow{2}{*}{\textbf{False Postive}} \\ \cline{3-4}
			&                                               & \textbf{Union Approach}     & \textbf{Convex Hull}     &                                         \\ \hline
			\textit{thttpd::expand\_symlinks}                  & 4,621                                         & 3,828 (82.84\%)             & 4,025(87.10\%)           & 0 (0.0\%)                               \\ \hline
			\textit{thttpd::defang}                            & 4,122                                         & 2,452 (59.53\%)             & 3,382 (82.12\%)          & 0 (0.0\%)                               \\ \hline
		\end{tabular}
	}
	\egroup
	\caption{ Comparison of redundant bounds checking bypassed by CHOP between Union and Convex Hull Approaches}
	\label{Comparision}
	\hspace{1mm}	
\end{table*}
%

\subsection{Memory overheads and code instrumentation due to {CHOP}}

We note that CHOP's memory overhead for storing Knowledge Base and additional code instrumentation are modest, since the Knowledge Base mainly stores the constructed safe region, which can be fully represented by data points as described in section~\ref{subsec:sdKB}. The safe region can be stored as a two-dimensional matrix, with each row representing one facet of the convex hull.
Storing such matrix is light-weight. Our experiments show that the worst memory overhead is only 20KB for the benchmarks we evaluated and the maximum code size increased is less than 5\% of the check-hotspot functions. Across all applications, CHOP has an average 7.3KB memory overhead with an average 2.1\% code increase. Overall, we reduce memory overhead by roughly 50\% compared to SoftBound memory requirements.

\begin{table*}[h]
	\vspace{-0.2in}
	\footnotesize
	\centering
	\bgroup
	\def\arraystretch{0.7}%
	\resizebox{1.0\textwidth}{!}{%
		\begin{tabular}{c|c|c|c}
			\cline{2-3}
			\multicolumn{1}{l|}{}                                      & \multicolumn{2}{c|}{\textbf{Time spent in}}              & \multicolumn{1}{l}{}                                   \\ \hline
			\multicolumn{1}{|c|}{\textbf{Function name}}               & \textbf{SoftBound} & \textbf{C.O.S} & \multicolumn{1}{l|}{\textbf{Bounds Check Time Reduction}} \\ \hline
			\multicolumn{1}{|c|}{\textit{bzip2::generateMTFValues}}           & 77.21s    & 39.46s                     & \multicolumn{1}{c|}{48.89\%}                  \\ \hline
			\multicolumn{1}{|c|}{\textit{bzip2::mainGtU}}                     & 47.94s    & 6.26s                      & \multicolumn{1}{c|}{86.94\%}                  \\ \hline
			\multicolumn{1}{|c|}{\textit{bzip2::BZ2\_decompress}}             & 35.58s    & 9.10s                      & \multicolumn{1}{c|}{74.42\%}                  \\ \hline
			\multicolumn{1}{|c|}{\textit{hmmer::P7Viterbi}}                   & 3701.11s  & 812.91s                    & \multicolumn{1}{c|}{78.04\%}                  \\ \hline
			\multicolumn{1}{|c|}{\textit{lbm::LBM\_performStreamCollide}}   & 1201.79s        & 407.06s                        & \multicolumn{1}{c|}{66.13\%}                        \\ \hline
			\multicolumn{1}{|c|}{\textit{sphinx3::vector\_gautbl\_eval\_logs3}} & 1580.03s  & 318.10s                    & \multicolumn{1}{c|}{79.87\%}                  \\ \hline
			\multicolumn{1}{|c|}{\textit{sphinx3::mgau\_eval}}                  & 1582.68s  &473.10s                   & \multicolumn{1}{c|}{70.11\%}                  \\ \hline
			\multicolumn{1}{|c|}{\textit{sphinx3::subvq\_mgau\_shortlist}}      & 270.84s   & 221.81s                    & \multicolumn{1}{c|}{18.1\%}                   \\ \hline
			\multicolumn{1}{|c|}{\textit{thttpd::httpd\_parse\_request}}      
			& 151.2s &121.0s               & \multicolumn{1}{c|}{95.80\%}       \\ \hline
			\multicolumn{1}{|c|}{\textit{thttpd::handle\_newconnect}}      
			& 40.4s   &12.9s              & \multicolumn{1}{c|}{73.52\%}        \\ \hline
		\end{tabular}
		
	}	
	\egroup
	\caption{Execution time of Check-HotSpot functions for SoftBound and C.O.S, and the resulting bounds check time reduction.}
	\label{Speedup}
\end{table*}

\begin{table*}[h]
	\vspace{-0.2in}
	\footnotesize
	\centering
	\bgroup
	\def\arraystretch{0.7}%
	\resizebox{1.0\textwidth}{!}{%
		\begin{tabular}{c|c|c|c|cc}
			\cline{2-4}
			\multicolumn{1}{l|}{}                                         & \multicolumn{3}{c|}{\textbf{Time spent in}}                                          & \multicolumn{1}{l}{}                         & \multicolumn{1}{l}{}                      \\ \hline
			\multicolumn{1}{|c|}{\multirow{2}{*}{\textbf{Function name}}} & \multirow{2}{*}{\textbf{SoftBound}} & \multicolumn{2}{c|}{\textbf{C.O.S}}            & \multicolumn{2}{c|}{\textbf{Bounds Check Time Reduction}}                                   \\ \cline{3-6} 
			\multicolumn{1}{|c|}{}                                        &                                     & \textbf{Union Approach} & \textbf{Convex Hull} & \multicolumn{1}{c|}{\textbf{Union Approach}} & \multicolumn{1}{c|}{\textbf{Convex Hull}} \\ \hline
			\multicolumn{1}{|c|}{\textit{thttpd::expand\_symlinks}}       & 19.6s                               & 5.30s                   & 3.70s                & \multicolumn{1}{c|}{72.91\%}                 & \multicolumn{1}{c|}{81.12\%}              \\ \hline
			\multicolumn{1}{|c|}{\textit{thttpd::defang}}                 & 5.37s                             & 2.44s                 & 1.21s              & \multicolumn{1}{c|}{54.56\%}                 & \multicolumn{1}{c|}{77.46\%}              \\ \hline
		\end{tabular}
	}	
	\egroup
	\caption{Comparison of Execution time of Check-HotSpot functions under SoftBound and CHOP between Union and Convex Hull Approaches.}
	\label{Speedup_Comparision}
	\vspace{-0.2in}
\end{table*}

\subsection{Execution time Overhead caused by CHOP}\label{extra}

CHOP bypasses redundant bounds checking by program profiling and safe region queries. We perform a comparison on the breakdown of execution time and found that the average overhead of a SoftBound check is 0.035s while the average of \textit{CHECK\_CHOP()}(together with trip count) overhead is 0.0019s.

\subsection{Case Studies}\label{CaseStudies}

In this section, we present detailed experimental results on the effectiveness of bounds check bypassing for both Union and Convex Hull approaches. Note that we only present the results from Convex Hull approach if the removal conditions are the same with Union approach. For those functions with different removal conditions, we further compare the performance between these two approaches. 
We also summarize the SoftBound overhead before and after redundant bounds checking bypassing using CHOP's Convex Hull approach, as well as the resulting execution time reduction, as shown in Table ~\ref{Speedup} and Table~\ref{Speedup_Comparision}.

\subsubsection{bzip2}
~\textit{bzip2} is a compression program to compress and decompress inputs files, such as TIFF image and source tar file. We identified three Check-HotSpot functions in~\textit{bzip2}: \textit{bzip2::mainGtU}, \textit{bzip2::generateMTFValues} and \textit{bzip2::BZ2\_decompress}. We use the function~\textit{bzip2::mainGtU} as an example to illustrate  how CHOP avoids redundant runtime checks in detail. Using Dependency Graph Construction from section~\ref{subsec:sdDG}, we fist identify \(nblock\), \(i_1\), and \(i_2\) as the pointer-affecting variables in~\textit{bzip2::mainGtU} function. For each execution, the Profile-guided Inference module computes and updates the Safe Region, which results in the following (sufficient) conditions for identifying redundant bounds checking in \textit{bzip2::mainGtU}: \[ nblock > i_1 +20\ or\ nblock > i_2 +20. \] Therefore, every time this Check-HotSpot function is called, CHOP will trigger runtime check bypassing if the inputs variables' values \(nblock\), \(i_1\), and \(i_2\) satisfy the conditions above. Because its Safe Region is one dimensional, the calculation of check bypassing conditions is indeed simple and only requires program input variables \(i_1\) and \(i_2\) (that are the array indexes) and \(nblock\) (that is the input array length). If satisfied, the conditions result in complete removal of bounds checking in function \textit{bzip2::mainGtU}. Our evaluation shows that it is able to eliminate over 99\% redundant checks.

For the second Check-HotSpot function~\textit{bzip2::generateMTFValue}, CHOP determines that array bounds checking could be bypassed for five different target pointers inside of the function. In this case, CHOP optimization reduces execution time overhead from 77.21s to 39.46s. We can see this number is near proportional to the number of checks removed by CHOP in Table~\ref{Checks}. 

The last Check-HotSpot function~\textit{bzip2::BZ2\_decompress} has over 200 lines of code. Similar to function~\textit{bzip2::generateMTFValue}, it also has five target pointers that share similar bounds check conditions. CHOP deploys a function-level removal for function~\textit{bzip2::BZ2\_decompress}. As we can see from Table~\ref{Speedup}, CHOP obtained a 74.42\% execution time reduction, which is consistent with the number of redundant bounds checking identified by CHOP presenting in Table~\ref{Checks}.

\subsubsection{hmmer}
\textit{hmmer} is a program for searching DNA gene sequences, which implements the~\textit{Profile Hidden Markov Models} algorithms and involves many double pointer operations. There is only one Check-HotSpot function, \textit{P7Viterbi}, which contributes over 98\% of SoftBound overhead. 

Inside of the function \textit{hmmer::P7Viterbi}, there are four double pointers: \textit{xmx, mmx, imx} and \textit{dmx}. To cope with double pointers in this function, we consider the row and column array bounds separately and construct a Safe Region for each dimension. Besides the 4 double pointers, we also construct conditions for check bypassing for another 14 pointers. The SoftBound overhead is significantly reduced from 3701.11s to 812.91s, rendering an execution time reduction of 78.94\% .

\subsubsection{lbm}
lbm is developed to simulate incompressible fluids in 3D and has only 1 Check-HotSpot function: \textit{lbm::LBM\_performStreamCollide}. The function has two pointers (as input variables) with pointer assignments and dereferencing inside of a for-loop. It suffers from high bounds check overhead in SoftBound, because pointer dereferencing occurs repeatedly inside the {\it for} loop, which results in frequent bounds checking. On the other hand, CHOP obtains the bounds check bypassing conditions for each pointer dereferencing. By further combining these conditions, we observed that the pointer dereferencing always access the same memory address, implying that it is always safe to remove all bounds checking in future executions after bounds checking are performed in the first execution. Thus, CHOP is able to bypass 100\% redundant checks which leads to an execution time reduction of 66.13\% .

\subsubsection{sphinx3}
Sphinx3 is a well-known speech recognition system, it is the third version of sphinx derived from sphinx2~\cite{huang1993sphinx}. 
The first Check-HotSpot function in Sphinx3 is \textit{sphinx3::vector\_gautbl\_eval\_logs3} and there are four target pointers inside this function. Due to the identical access pattern, once we derive the bounds check bypassing conditions for one single pointer, it also applies to all the others, allowing all redundant checks to be bypassed simultaneously in this function. As shows in Table~\ref{Checks}, CHOP bypass 100\% of redundant checks with a resulting execution time of 318.10s, which achieves the optimal performance.

We observed a similar behavior for the second Check-HotSpot function in Sphinx3: ~\textit{sphinx3::mgau\_eval}. CHOP achieves 100\% redundant bounds check bypassing with an execution time reduction of 70.11\%, from 1582.68s in SoftBound to 473.10s after CHOP's redundant bounds check bypassing.

The last function~\textit{sphinx3::subvq\_mgau\_shortlist} also has four target pointers. CHOP optimized SoftBound incurs an overhead of 221.81s, when the original SoftBound overhead is 270.84s. For this function, CHOP only removed 18.1\% redundant checks, which is the lowest in our evaluations.The reason is that the pointer \textit{vqdist} inside this function has indirect memory access, that its index variable is another pointer ~\textit{map}. The dependency graph we constructed cannot represent the indirect memory access relation between these two pointers. Since CHOP is not able to handle pointers that perform indirect memory accesses, it only removes about 18\% of the bounds checking. We note that capturing such memory access dependencies is possible via extending our dependency graph to model complete memory referencing relations. We will consider this as future work. 

\subsubsection{thttpd}\label{thttpd}
thttpd is a light-weight HTTP server. A buffer overflow vulnerability has been exploited in thttpd 2.2x versions within a function called \textit{thttpd::defang()}, which 
replaces the special characters "$<$" and "$>$" in url \textit{str} with "\&lt;" and "\&gt;" respectively, then outputs the new string as \textit{dfstr}. The function \textit{thttpd::defang()} can cause an buffer overflow when the length of the new string is larger than 1000. To evaluate the performance of CHOP, we generate 1000 thttpd requests with random URL containing such special characters, and submit the requests to a host thttpd server to trigger \textit{thttpd::defang()}.

CHOP's dependency analysis reveals that the pointer \textit{dfstr} has two pointer-affecting variables \textit{s} and \textit{n}, where $s$ denotes the total number of special characters in the input url and $n$ denotes the length of the input url. The bound of \textit{dfstr} is a constant value of 1000. To illustrate the safe region construction using Convex Hull and Union approaches, we consider the first two input data points from two executions: $(s_1,n_1)=(1,855)$ and $(s_2,n_2)=(16,60)$.
It is easy to see that based on the two input data points, the safe region built by Union approach is $SR_{union}=SR(1) \cup SR(2)$, where $ SR(1)$ = $\{(s, n): 0 \le s \le 1, 0 \le n \le 855\}$ and $\{(s, n): 0 \le s \le 16, 0 \le n \le 60\} $ are derived from the two data points, respectively. On the other hand, our Convex Hull approach extends $SR_{union}$ into a convex hull with linear boundaries. It results in a safe region $SR_{convex}$ = $\{(s, n): 0 \le s \le 16, 0 \le n \le 855, 53*s + n \le 908\}$. As the ground truth, manual analysis shows that the sufficient and necessary condition for safe bounds checking removal is given by an optimal safe region $SR_{opt}=\{(s,n): \ 3*s + n + 1< = 1000\}$. While more input data points are needed to achieve $SR_{opt}$, the safe region constructed using Convex Hull approach significantly improves that of Union approach, under the assumption that the pointer-affecting variables are linearly related to target pointers and array bounds.
With 10 randomly generated thttpd requests, we show in Table~\ref{Comparision} that CHOP's Convex Hull approach successfully bypasses 82.12\% redundant bounds checking with 0 false positive, whereas Union approach in Table~\ref{Comparision} and Table~\ref{Speedup_Comparision} is able to bypass 59.53\% redundant bounds checking.
Furthermore, Union approach has 14.89\% runtime overhead, compared to 21.90\% for Convex Hull approach. This is because the runtime redundant checks identification in Union approach only requires (component-wise) comparison of an input vector and corner points. On the other hand, Convex Hull approach needs to check whether the input vector falls inside a convex safe region, which requires checking all linear boundary conditions and results in higher runtime overhead. Additionally, CHOP bypass 100\% runtime bounds checking in function \textit{thttpd::httpd\_parse\_request} and 73.52\% checks in function \textit{thttpd::handle\_newconnect}.


\subsubsection{Others}\label{sec:case-others}
We also have some interesting results due to imperfect profiling. For example, in the application \textit{libquantum}, the Check-HotSpot functions are small and we only identify 3 linear assignment and 3 non-linear assignment of related pointers. Thus the convex hull approach will be ineffective. As shown in Fig.~\ref{fig:cmp-ratio}, we can only remove less than 50\% runtime checks when WPBound can remove much more. Similarly, in the application \textit{milc}, a large portion of the pointer-related assignments have multiplication and convex hull-based approach cannot deal with non-linear relationships among pointer-affecting variables. Additionally, in the application \textit{lbm}, due to the intensive use of macros for preprocessor, our static code parsing tool cannot recognize the complete function bodies. As a result, WPBound outperforms CHOP on the ratio of dynamic bounds check reduction and execution overhead reduction.

\section{Discussion}
We discuss some limitations of our current design in this section.

\textbf{CHOP currently is built to optimize SoftBound.} Since CHOP is based on SoftBound (which is built upon LLVM), it currently only works on programs compiled with Clang. We note that this research framework can be easily extended to other environments with engineering effort.

\textbf{The test programs need to be profiled.} In order to find the Check-HotSpot functions, we have to compile the programs with and without SoftBound. Also, test programs are executed several times to initialize the Safe Region. However, this Safe Region initialization is a one-time effort and will not need to be repeated.

\textbf{The performance of convex hull approach could drop when the dimension of data points gets too high.} As shown in our evaluation, the time overhead of convex hull approach could be higher than that of union approach. The higher the dimension of convex hull is, the more facets the safe region has and the more comparisons it will need to decide if a new data point is in the convex hull. However, since the query to convex hull only needs to occur once per function (in function-level check bypassing), it still reduce the runtime overhead of SoftBound, and it will provide high ratio of check bypassing compared against Union approach.

\textbf{The convex hull approach is built on an assumption.} In order to build the convex hull as Safe Region, CHOP will require the relationship among the pointer-affecting variables to be linear. We show the ratio of linear assignments from the applications in Table~\ref{tab:linear}. The examples of linear and non-linear assignments related to pointers include but are not limited to the following:

\textbf{Linear assignments:}
\begin{itemize}
	\item cand\_x=offset\_x+spiral\_search\_x[pos]
	\item v=s$\rightarrow$selectorMtf[i]
\end{itemize}

\textbf{Non-linear assignments:}
\begin{itemize}
	\item i1=sizeof(block)/sizeof(block[0])
	\item int max\_pos=(2$\cdot$search\_range+1)$\cdot$(2$\cdot$search\_range+1);
\end{itemize}

We observe that most applications from SPEC2006 have high ratios of linear pointer-related assignments. The ratio of non-linear assignments is higher in some applications such as \textit{lbm}, where we found that it intensively uses macros. The assignments with macros and functions calls will be classified as non-linear assignments by our algorithm. In the function \textit{srcGrid} from \textit{lbm}, the macro \textit{SRC\_ST(srcGrid)} performs certain calculation based on the grid index and value from \textit{srcGrid}.

\begin{table}[]
	\centering
	\begin{tabular}{|c|c|c|c|}
		
		\hline
		\multirow{2}{*}{Application} & \multicolumn{2}{c|}{Related Assignments} & \multirow{2}{*}{\begin{tabular}[c]{@{}c@{}}Ratio of\\ Linear Assignments(\%)\end{tabular}} \\ \cline{2-3}
		& Linear            & Non-linear           &                                                                                            \\ \hline
		Bzip2                        & 1174              & 60                   & 95.1                                                                                       \\ \hline
		lbm                          & 58                & 40                   & 59.2                                                                                       \\ \hline
		sphinx3                      & 342               & 45                   & 88.4                                                                                       \\ \hline
		hmmer                        & 687               & 6                    & 99.1                                                                                       \\ \hline
		h264ref                      & 298               & 11                   & 96.4                                                                                       \\ \hline
		libquantum                   & 3                 & 3                    & 50                                                                                         \\ \hline
		milc                         & 162               & 78                   & 67.5                                                                                       \\ \hline
		Total                        & 2724              & 243                  & 91.8                                                                                       \\ \hline
	\end{tabular}
	\caption{Number of linear and non-linear assignments on Check-HotSpot functions from SPEC2006 applications.}
	\label{tab:linear}
	\vspace{-0.2in}
\end{table}

\textbf{CHOP does not perform bounds check bypassing beyond the function level}. We will consider the inter-procedural analysis in our future works to detect redundant bounds check across the whole program~\cite{jensen2010interprocedural,Sui:2016:SIS:2892208.2892235}.
\section{Related Work}
\label{sec:related}

C and C++ are unsafe programming languages and plenty of efforts have been made towards securing the memory usages of C/C++ programs~\cite{song2018sok}. Memory safety and usage analysis have been studied widely~\cite{chen_jetc14, venkataramani2011deft, chen2012repram, oh2011lime}. Some existing works try to find memory-related vulnerabilities in source code or IR (during compilation) by direct static analysis~\cite{ganapathy2003buffer,splint,Chucky_CCS_2013,Pscan,dor2003cssv}. For example, splint~\cite{splint} utilizes light-weight static analysis to detect bugs in annotated programs, by checking if the properties of the program are consistent with the annotation. 
Yamaguchi et. al.~\cite{Chucky_CCS_2013,AST_ACSA_2012} use machine learning techniques to identify the similarity of code patterns to facilitate discovery of vulnerabilities. Clone-Hunter~\cite{xue2018clone} and Clone-Slicer~\cite{xue2018clone} aim to detect code clones in program binaries for accelerated bounds check removal. Nurit et. al.~\cite{dor2003cssv} target string-related bugs in C program with a conservative pointer analysis using abstracted constraint expressions for pointer operations similar to ABCD~\cite{ABCD_SIGPLAN_2000}.

While such static techniques can be quite scalable and low-cost (with no impact to runtime overhead), it often result in incomplete and inaccurate analysis. Pure static code analysis could suffer from undecidable factors that can only be known during runtime. Hence, some works build safety rules based on source code or compile time analysis, then enforce such rules during runtime to prevent undesired behaviors such as out-of-bound accesses or unauthorized control flow transfers~\cite{CCured_SIGPLAN_2002, shen2006tradeoffs, WIT_2008,li2016sarre,SoftBound_PLDI_2009,chen2017damgate, memtracker, doudalis2012effective}. Necula et. al. propose CCured~\cite{CCured_SIGPLAN_2002}, which is a type safe system that leverages rule-based type inference to determine ``safe” pointers. It categories pointers into three types \(\{safe, seq, dynamic\}\) then applies different checking rules for them. Akritidis et.al~\cite{WIT_2008} perform points-to analysis during compile time to mark the objects that can be written to, then prevent writes to unmarked objects during runtime. They also enforce the control transfers according to a pre-built CFG. As mentioned previously, SoftBound works in a similar way. It is built upon LLVM to track pointer metadata and perform bounds check when pointers are dereferenced. 

Such approaches typically instrument the program to insert customized checks which will be activated during runtime. Hence, the performance could be a serious issue due to the additional checks and metadata operations. Techniques that remove redundant checks to boost the runtime performance have been studied~\cite{ABCD_SIGPLAN_2000,wpcheck,statsym_dsn,Java_Hotspot,SIMBER}. WPBound statically analyzes the ranges of pointer values inside loops. During runtime, it compares such ranges with the actual runtime values obtained from SoftBound to determine if the bounds check can be removed from the loops. W{\"u}rthinger et. al.~\cite{Java_Hotspot} eliminate the bounds check based on the static analysis upon JIT IR during program compiling and keep a condition for every instruction that computes an integer value. Different from these works, SIMBER~\cite{SIMBER} and CHOP utilizes historical runtime data and profile-guided inferences to perform bounds check. However, SIMBER only uses union approach to construct safe region for bounds check bypassing and suffers from low bounds check bypassing rate in large-scale applications with multiple pointer-affecting variables.

\vspace{-.1in}
\section{Conclusion}
\label{sec:concl}

In this paper, we propose CHOP, a framework integrates profile-guided inference with spatial memory safety checks to perform redundant bounds check bypassing through Convex Hull Optimization. CHOP targets frequently executed functions instrumented by SoftBound that have redundant bounds checking. Our experimental evaluation on two different sets of real-world benchmarks shows that CHOP can obtain an average 71.29\% reduction in array bounds checking and 66.31\% reduction in bounds check execution time. 

\vspace{-.15in}
\section*{Acknowledgments}
This work was supported by the US Office of Naval Research
Awards N00014-17-1-2786 and N00014-15-1-2210. 
\bibliographystyle{IEEEtran}
\vspace{-0.3in}
\bibliography{chop} 	

\begin{thebibliography}{10}
\providecommand{\url}[1]{#1}
\csname url@samestyle\endcsname
\providecommand{\newblock}{\relax}
\providecommand{\bibinfo}[2]{#2}
\providecommand{\BIBentrySTDinterwordspacing}{\spaceskip=0pt\relax}
\providecommand{\BIBentryALTinterwordstretchfactor}{4}
\providecommand{\BIBentryALTinterwordspacing}{\spaceskip=\fontdimen2\font plus
\BIBentryALTinterwordstretchfactor\fontdimen3\font minus
  \fontdimen4\font\relax}
\providecommand{\BIBforeignlanguage}[2]{{%
\expandafter\ifx\csname l@#1\endcsname\relax
\typeout{** WARNING: IEEEtran.bst: No hyphenation pattern has been}%
\typeout{** loaded for the language `#1'. Using the pattern for}%
\typeout{** the default language instead.}%
\else
\language=\csname l@#1\endcsname
\fi
#2}}
\providecommand{\BIBdecl}{\relax}
\BIBdecl

\bibitem{glibc}
Google, ``Cve-2015-7547: glibc getaddrinfo stack-based buffer overflow,''
  \url{https://security.googleblog.com/2016/02/cve-2015-7547-glibc-getaddrinfo-stack.html}.

\bibitem{Cisco}
Cisco, ``Cve-2016-1287: Cisco asa software ikev1 and ikev2 buffer overflow
  vulnerability,''
  \url{https://tools.cisco.com/security/center/content/CiscoSecurityAdvisory/cisco-sa-20160210-asa-ike}.

\bibitem{SoftBound_PLDI_2009}
S.~Nagarakatte, J.~Zhao, M.~M. Martin, and S.~Zdancewic, ``Softbound: Highly
  compatible and complete spatial memory safety for c,'' vol.~44, no.~6.\hskip
  1em plus 0.5em minus 0.4em\relax ACM, 2009, pp. 245--258.

\bibitem{ABCD_SIGPLAN_2000}
R.~Bod{\'\i}k, R.~Gupta, and V.~Sarkar, ``Abcd: eliminating array bounds
  checking on demand,'' in \emph{ACM SIGPLAN Notices}, vol.~35, no.~5.\hskip
  1em plus 0.5em minus 0.4em\relax ACM, 2000, pp. 321--333.

\bibitem{wpcheck}
Y.~Sui, D.~Ye, Y.~Su, and J.~Xue, ``Eliminating redundant bounds checking in
  dynamic buffer overflow detection using weakest preconditions,'' \emph{IEEE
  Transactions on Reliability}, vol.~65, no.~4, 2016.

\bibitem{SIMBER}
H.~Xue, Y.~Chen, F.~Yao, Y.~Li, T.~Lan, and G.~Venkataramani, ``Simber:
  Eliminating redundant memory bound checks via statistical inference,'' in
  \emph{Proceedings of the IFIP International Conference on Computer
  Security}.\hskip 1em plus 0.5em minus 0.4em\relax Springer, 2017.

\bibitem{spec}
``Spec cpu 2006,'' \url{https://www.spec.org/cpu2006/}.

\bibitem{quickhull}
\BIBentryALTinterwordspacing
C.~B. Barber, D.~P. Dobkin, and H.~Huhdanpaa, ``The quickhull algorithm for
  convex hulls,'' \emph{ACM Trans. Math. Softw.}, vol.~22, no.~4, pp. 469--483,
  Dec. 1996. [Online]. Available:
  \url{http://doi.acm.org/10.1145/235815.235821}
\BIBentrySTDinterwordspacing

\bibitem{perf}
A.~C. de~Melo, ``The new linux’perf’tools,'' in \emph{Slides from Linux
  Kongress}, vol.~18, 2010.

\bibitem{Pscan}
D.~Alan, ``Pscan: A limited problem scanner for c source files,''
  \url{http://deployingradius.com/pscan/}.

\bibitem{splint}
D.~Evans and D.~Larochelle, ``Improving security using extensible lightweight
  static analysis,'' \emph{IEEE software}, vol.~19, no.~1, 2002.

\bibitem{PR-Miner}
Z.~Li and Y.~Zhou, ``Pr-miner: automatically extracting implicit programming
  rules and detecting violations in large software code,'' in \emph{ACM SIGSOFT
  Software Engineering Notes}, vol.~30, no.~5.\hskip 1em plus 0.5em minus
  0.4em\relax ACM, 2005.

\bibitem{srivastava2011security}
V.~Srivastava, M.~D. Bond, K.~S. McKinley, and V.~Shmatikov, ``A security
  policy oracle: Detecting security holes using multiple api implementations,''
  \emph{ACM SIGPLAN Notices}, vol.~46, no.~6, 2011.

\bibitem{AST_ACSA_2012}
F.~Yamaguchi, M.~Lottmann, and K.~Rieck, ``Generalized vulnerability
  extrapolation using abstract syntax trees,'' in \emph{Annual Computer
  Security Applications Conference}, 2012.

\bibitem{joern}
F.~Yamaguchi, ``Joern: A robust code analysis platform for c/c++,''
  \url{http://www.mlsec.org/joern/}.

\bibitem{yamaguchi2014modeling}
F.~Yamaguchi, N.~Golde, D.~Arp, and K.~Rieck, ``Modeling and discovering
  vulnerabilities with code property graphs,'' in \emph{2014 IEEE Symposium on
  Security and Privacy}.\hskip 1em plus 0.5em minus 0.4em\relax IEEE, 2014, pp.
  590--604.

\bibitem{yamaguchi2015automatic}
F.~Yamaguchi, A.~Maier, H.~Gascon, and K.~Rieck, ``Automatic inference of
  search patterns for taint-style vulnerabilities,'' in \emph{2015 IEEE
  Symposium on Security and Privacy}.\hskip 1em plus 0.5em minus 0.4em\relax
  IEEE, 2015, pp. 797--812.

\bibitem{owens2010sqlite}
M.~Owens and G.~Allen, \emph{SQLite}.\hskip 1em plus 0.5em minus 0.4em\relax
  Springer, 2010.

\bibitem{mcnally2012fuzzing}
R.~McNally, K.~Yiu, D.~Grove, and D.~Gerhardy, ``Fuzzing: the state of the
  art,'' DTIC Document, Tech. Rep., 2012.

\bibitem{woo2013scheduling}
M.~Woo, S.~K. Cha, S.~Gottlieb, and D.~Brumley, ``Scheduling black-box
  mutational fuzzing,'' in \emph{ACM SIGSAC conference on Computer \&
  communications security}, 2013.

\bibitem{huang1993sphinx}
X.~Huang, F.~Alleva, H.-W. Hon, M.-Y. Hwang, K.-F. Lee, and R.~Rosenfeld, ``The
  sphinx-ii speech recognition system: an overview,'' \emph{Computer Speech \&
  Language}, vol.~7, no.~2, pp. 137--148, 1993.

\bibitem{jensen2010interprocedural}
S.~H. Jensen, A.~M{\o}ller, and P.~Thiemann, ``Interprocedural analysis with
  lazy propagation,'' in \emph{International Static Analysis Symposium}.\hskip
  1em plus 0.5em minus 0.4em\relax Springer, 2010, pp. 320--339.

\bibitem{Sui:2016:SIS:2892208.2892235}
\BIBentryALTinterwordspacing
Y.~Sui and J.~Xue, ``Svf: Interprocedural static value-flow analysis in llvm,''
  in \emph{Proceedings of the 25th International Conference on Compiler
  Construction}, ser. CC 2016.\hskip 1em plus 0.5em minus 0.4em\relax New York,
  NY, USA: ACM, 2016, pp. 265--266. [Online]. Available:
  \url{http://doi.acm.org/10.1145/2892208.2892235}
\BIBentrySTDinterwordspacing

\bibitem{song2018sok}
D.~Song, J.~Lettner, P.~Rajasekaran, Y.~Na, S.~Volckaert, P.~Larsen, and
  M.~Franz, ``Sok: sanitizing for security,'' \emph{arXiv preprint
  arXiv:1806.04355}, 2018.

\bibitem{chen_jetc14}
J.~Chen, G.~Venkataramani, and H.~H. Huang, ``Exploring dynamic redundancy to
  resuscitate faulty pcm blocks,'' \emph{J. Emerg. Technol. Comput. Syst.},
  vol.~10, no.~4, Jun. 2014.

\bibitem{venkataramani2011deft}
G.~Venkataramani, C.~J. Hughes, S.~Kumar, and M.~Prvulovic, ``Deft: Design
  space exploration for on-the-fly detection of coherence misses,'' \emph{ACM
  Transactions on Architecture and Code Optimization (TACO)}, vol.~8, no.~2,
  p.~8, 2011.

\bibitem{chen2012repram}
J.~Chen, G.~Venkataramani, and H.~H. Huang, ``Repram: Re-cycling pram faulty
  blocks for extended lifetime,'' in \emph{IEEE/IFIP International Conference
  on Dependable Systems and Networks}, 2012.

\bibitem{oh2011lime}
J.~Oh, C.~J. Hughes, G.~Venkataramani, and M.~Prvulovic, ``Lime: A framework
  for debugging load imbalance in multi-threaded execution,'' in \emph{33rd
  International Conference on Software Engineering}.\hskip 1em plus 0.5em minus
  0.4em\relax ACM, 2011.

\bibitem{ganapathy2003buffer}
V.~Ganapathy, S.~Jha, D.~Chandler, D.~Melski, and D.~Vitek, ``Buffer overrun
  detection using linear programming and static analysis,'' in
  \emph{Proceedings of the 10th ACM conference on Computer and communications
  security}.\hskip 1em plus 0.5em minus 0.4em\relax ACM, 2003, pp. 345--354.

\bibitem{Chucky_CCS_2013}
F.~Yamaguchi, C.~Wressnegger, H.~Gascon, and K.~Rieck, ``Chucky: Exposing
  missing checks in source code for vulnerability discovery,'' in \emph{ACM
  Conference on Computer \& Communications security}, 2013.

\bibitem{dor2003cssv}
N.~Dor, M.~Rodeh, and M.~Sagiv, ``Cssv: Towards a realistic tool for statically
  detecting all buffer overflows in c,'' in \emph{ACM Sigplan Notices},
  vol.~38, no.~5.\hskip 1em plus 0.5em minus 0.4em\relax ACM, 2003, pp.
  155--167.

\bibitem{xue2018clone}
H.~Xue, G.~Venkataramani, and T.~Lan, ``Clone-hunter: accelerated bound checks
  elimination via binary code clone detection,'' in \emph{Proceedings of the
  2nd ACM SIGPLAN International Workshop on Machine Learning and Programming
  Languages}.\hskip 1em plus 0.5em minus 0.4em\relax ACM, 2018, pp. 11--19.

\bibitem{CCured_SIGPLAN_2002}
G.~C. Necula, S.~McPeak, and W.~Weimer, ``Ccured: Type-safe retrofitting of
  legacy code,'' vol.~37, no.~1, pp. 128--139, 2002.

\bibitem{shen2006tradeoffs}
J.~Shen, G.~Venkataramani, and M.~Prvulovic, ``Tradeoffs in fine-grained heap
  memory protection,'' in \emph{Proceedings of the 1st workshop on
  Architectural and system support for improving software dependability}.\hskip
  1em plus 0.5em minus 0.4em\relax ACM, 2006, pp. 52--57.

\bibitem{WIT_2008}
P.~Akritidis, C.~Cadar, C.~Raiciu, M.~Costa, and M.~Castro, ``Preventing memory
  error exploits with wit,'' in \emph{2008 IEEE Symposium on Security and
  Privacy (sp 2008)}.\hskip 1em plus 0.5em minus 0.4em\relax IEEE, 2008, pp.
  263--277.

\bibitem{li2016sarre}
Y.~Li, F.~Yao, T.~Lan, and G.~Venkataramani, ``Sarre: semantics-aware rule
  recommendation and enforcement for event paths on android,'' \emph{IEEE Tran.
  on Information Forensics and Security}, vol.~11, no.~12, 2016.

\bibitem{chen2017damgate}
Y.~Chen, T.~Lan, and G.~Venkataramani, ``Damgate: Dynamic adaptive
  multi-feature gating in program binaries,'' in \emph{Proceedings of the 2017
  Workshop on Forming an Ecosystem Around Software Transformation}.\hskip 1em
  plus 0.5em minus 0.4em\relax ACM, 2017, pp. 23--29.

\bibitem{memtracker}
G.~Venkataramani, I.~Doudalis, Y.~Solihin, and M.~Prvulovic, ``Memtracker: An
  accelerator for memory debugging and monitoring,'' \emph{ACM Transactions on
  Architecture and Code Optimization (TACO)}, vol.~6, no.~2, p.~5, 2009.

\bibitem{doudalis2012effective}
I.~Doudalis, J.~Clause, G.~Venkataramani, M.~Prvulovic, and A.~Orso,
  ``Effective and efficient memory protection using dynamic tainting,''
  \emph{IEEE Transactions on Computers}, vol.~61, no.~1, pp. 87--100, 2012.

\bibitem{statsym_dsn}
F.~Yao, Y.~Li, Y.~Chen, H.~Xue, G.~Venkataramani, and T.~Lan, ``{StatSym:
  Vulnerable Path Discovery through Statistics-guided Symbolic Execution},'' in
  \emph{IEEE/IFIP International Conference on Dependable Systems and Networks},
  2017.

\bibitem{Java_Hotspot}
T.~W{\"u}rthinger, C.~Wimmer, and H.~M{\"o}ssenb{\"o}ck, ``Array bounds check
  elimination for the java hotspot client compiler,'' in \emph{Proceedings of
  the 5th international symposium on Principles and practice of programming in
  Java}.\hskip 1em plus 0.5em minus 0.4em\relax ACM, 2007, pp. 125--133.

\end{thebibliography}

\end{document}